\documentclass[10pt]{article}

\usepackage[protrusion=false,expansion,final,tracking=true,kerning=true,spacing=true,factor=1100,stretch=10,shrink=10]{microtype}
\usepackage[utf8]{inputenc}
\usepackage[colorlinks=true]{hyperref}
\usepackage[T1]{fontenc}

\usepackage{amsmath}
\usepackage{amsfonts}
\usepackage{amssymb}
\usepackage{amsthm}
\usepackage{mathtools}
\usepackage[dvipsnames]{xcolor}
\usepackage{xifthen}
\usepackage{tikz}
\usetikzlibrary{math,calc}
\usepackage{graphicx}
\DeclareGraphicsExtensions{.png,.pdf,.jpg}
\graphicspath{{pics/}}
\usepackage[export]{adjustbox}
\usepackage{authblk}
\usepackage{orcidlink}
\usepackage{caption}
\newcommand{\calA}{\mathcal{A}}
\newcommand{\calB}{\mathcal{B}}

\newcommand{\calD}{\mathcal{D}}

\newcommand{\calT}{\mathcal{T}}

\newcommand{\calV}{\mathcal{V}}

\newcommand{\frm}{\mathfrak{m}}

\newcommand{\bbE}{\mathbb{E}}

\newcommand{\bbP}{\mathbb{P}}
\newcommand{\bbQ}{\mathbb{Q}}
\newcommand{\bbR}{\mathbb{R}}
\newcommand{\bbS}{\mathbb{S}}

\newcommand{\bbV}{\mathbb{V}}

\newcommand{\bbZ}{\mathbb{Z}}
\newcommand{\sfb}{\mathsf{b}}

\newcommand{\sfo}{\mathsf{o}}

\newcommand{\sfr}{\mathsf{r}}

\newcommand{\sfD}{\mathsf{D}}

\newcommand{\sfK}{\mathsf{K}}

\newcommand{\sfO}{\mathsf{O}}

\newcommand{\sfW}{\mathsf{W}}

\newcommand{\sfZ}{\mathsf{Z}}

\DeclarePairedDelimiterX{\setof}[2]{\{}{\}}{%
  #1 \,\delimsize|\, #2%
}
\DeclarePairedDelimiter\abs{\lvert}{\rvert}%
\DeclarePairedDelimiter\norm{\lVert}{\rVert}%
\DeclarePairedDelimiter\angularbrackets{\langle}{\rangle}%
\newcommand*{\scalesub}[1]{%
  \mathchoice
    {\raisebox{-0.2ex}{\scalebox{0.9}{\(\displaystyle #1\)}}}%      Display style
    {\raisebox{-0.2ex}{\scalebox{0.9}{\(\textstyle #1\)}}}%         Text style
    {\raisebox{-0.1ex}{\scalebox{0.9}{\(\scriptstyle #1\)}}}%        Script style
    {\raisebox{-0.1ex}{\scalebox{0.9}{\(\scriptscriptstyle #1\)}}}%  Script-script style
}

\newcommand{\myQED}{\hfill\ensuremath{\qedsymbol}}

\newcommand{\normI}[1]{\norm{#1}_{\mkern-1mu\scalesub{1}}}
\newcommand{\normII}[1]{\norm{#1}_{\mkern-1mu\scalesub{2}}}
\newcommand{\normsup}[1]{\norm{#1}_{\mkern-1mu\scalesub{\infty}}}
\newcommand{\defby}{\coloneqq}%{\stackrel{\mathrm{def}}{=}}
\newcommand{\bydef}{\eqqcolon}%{\stackrel{\mathrm{def}}{=}}
\newcommand{\IF}[1]{\boldsymbol{1}_{\{#1\}}}
\newcommand{\spr}[3][]{\angularbrackets[#1]{#2, #3}}%

\newcommand{\dd}{\mathrm{d}}
\renewcommand{\emptyset}{\varnothing}
\newcommand{\arsinh}{\operatorname{arsinh}}

\newcommand{\artanh}{\operatorname{artanh}}
\newcommand{\loct}{\ell}
\newcommand{\prob}{\bbP}
\newcommand{\expectation}{\bbE}
\newcommand{\variance}{\bbV}
\newcommand{\p}[2]{\prob^{#1}_{#2}}
\newcommand{\E}[2]{\expectation^{#1}_{#2}}
\newcommand{\V}[2]{\variance^{#1}_{#2}}
\newcommand{\Z}[2]{\sfZ^{#1}_{#2}}
\newcommand{\D}[2]{\sfD^{#1}_{#2}}

\newcommand{\lambdac}{\lambda_{\mathrm{c}}}

\newcommand{\Bl}{\sfb_\lambda}
\newcommand{\W}{\sfW}
\newcommand{\icl}{\nu}
\newcommand{\fe}{f}
\newcommand{\K}{\sfK}
\newcommand{\Kl}{\K_\lambda}

\newcommand{\concat}{\circ}
\newcommand{\fcone}{Y^\blacktriangleleft}
\newcommand{\fbone}{Y^\blacktriangleright}

\newcommand{\QL}{\bbQ_{\mathrm L}^{h,\lambda}}
\newcommand{\QR}{\bbQ_{\mathrm R}^{h,\lambda}}
\newcommand{\QD}{\bbQ^{h,\lambda}}
\newcommand{\Skel}[1]{\operatorname{Skel}(#1)}
\newcommand{\bfv}{\boldsymbol{v}}

\newcommand{\bftau}{\boldsymbol{\tau}}
\newcommand{\MP}{\sfr} % Notation for the macroscopic paths
\newcommand{\Leb}{\mathrm{Leb}}
\newcommand{\Hess}[1]{H_{#1}}
\theoremstyle{plain}
\newtheorem{theorem}{Theorem}[section]
\newtheorem{lemma}[theorem]{Lemma}
\newtheorem{proposition}[theorem]{Proposition}
\newtheorem{corollary}[theorem]{Corollary}

\theoremstyle{definition}

\newtheorem{remark}{Remark}[section]
\AtBeginEnvironment{remark}{%
  \pushQED{\qed}\renewcommand{\qedsymbol}{\(\circ\)}%
}
\AtEndEnvironment{remark}{\popQED}

\title{Typical geometry of self-repelling polymers in a constant force field}
\author{Kamil Khettabi\thanks{Kamil.Khettabi@unige.ch} }
\author{Yvan Velenik\thanks{Yvan.Velenik@unige.ch} \orcidlink{0000-0002-6459-3197} }
\affil{Section de mathématiques, Université de Genève}

\begin{document}

\maketitle

\begin{abstract}
  We study a general class of self-repelling polymers on \(\mathbb Z^2\), including the simple random walk, the self-avoiding walk and the repulsive Domb--Joyce model, in the presence of a constant force field acting on each monomer. Conditioning the polymer to have fixed length and fixed endpoints, we identify the limiting free energy and prove that typical trajectories concentrate exponentially near a deterministic macroscopic shape. This shape is characterized as the unique minimizer of a variational problem and can be interpreted as a geodesic of a height-dependent Finsler metric. We also analyze two limiting regimes with universal features: for small field strength, in the symmetric case, the geodesic is close to a classical catenary, while for large field strength it converges to a universal polygonal shape governed by the nearest-neighbor lattice constraint.
\end{abstract}

% \yvan{Open problems (if we decide to put some): Is the minimizer always convex? How universal is the behavior at small \(g\) when \(a\neq 0\)? Extension to a suitable large deviation principle?  What is the process describing fluctuations around the minimizer (not obvious as there can be fluctuations both spatial and in monomer distribution)? What happens in higher dimensions?}

%%%%%%%%%%%%%%%%%%%%%%%
\section{Introduction}
\label{sec:Intro}
%%%%%%%%%%%%%%%%%%%%%%%

Consider a long polymer chain, anchored at the origin and subject to a tensile force acting on its free end. The analysis of the effect of this force on the chain conformation, and in particular on its degree of extension, has been a focal point of research over the last few decades, driven largely by the development of single-molecule experimental techniques. This has also led to a number of mathematically rigorous works addressing these issues; see, for instance, \cite{vanRensburg+Whittington-2013, Beaton-2015, Ioffe+Velenik-2008, Ioffe+Velenik-2012}. In particular, \cite{Ioffe+Velenik-2008, Ioffe+Velenik-2012} analyze in detail a general class of polymer models with either attractive or repulsive self-interactions. They derive a full local limit theorem for the position of the free endpoint as well as for the statistics of local patterns, and prove that the transition from a collapsed phase to an extended phase is first order for self-attractive polymers.

In the present paper, we analyze a related problem: that of a polymer in a constant force field, such as a gravitational field, an electric field acting on a charged polymer, or, as a first approximation, a polymer in an extensional flow. Namely, we consider a general class of self-repelling polymer models on \(\bbZ^2\) with a constant force acting on each monomer. Assuming that both the length of the polymer and the positions of its endpoints are fixed, we investigate the typical conformation of the polymer. Our main results are the determination of the associated free energy, as well as a proof of concentration of typical realizations on the minimizer of a suitable variational problem.

The presence of a force acting on every monomer makes the problem quite different from the usual endpoint-pulling setup, with additional difficulties arising from the lack of translation invariance. The energetic contribution of the field depends on the whole trajectory, while the fixed-length constraint compels the polymer to distribute its microscopic length along the macroscopic curve. The resulting variational problem does not split into an elastic term plus an independent gravitational potential. Rather, the field modifies the local effective tension itself, leading to a height-dependent Finsler metric.

A central point of the paper is that this apparently complicated variational problem has a hidden convex structure. After passing to a dual formulation and reparameterizing the curve by monomer time, we construct the minimizer explicitly by a shooting argument and show it to be unique and stable. The case of the simple random walk is exactly solvable, and the geodesic can be determined explicitly. We also show that the small field regime leads to universal behavior, at least in the symmetric case where the two endpoints are pinned at the same height: the geodesic is closely approximated by a catenary of the same apparent length. Finally, in the large field regime, a restricted form of universality also applies, where the limiting geometry of the geodesic is dominated by lattice effects.

\medskip
In the remainder of this section, we introduce the relevant polymer models (Section~\ref{sec:Intro:PathEnsembles}) and the associated thermodynamic quantities (Section~\ref{sec:Intro:Thermo}). Our main results are presented in Section~\ref{sec:Intro:MainResults}. Sections~\ref{sec:Thermo} to~\ref{sec:ProofMain} are devoted to the proofs. A more detailed roadmap to the paper can be found in Section~\ref{sec:Intro:Roadmap}.

\paragraph{Notation conventions.}
For functions of a generic vector variable in \(\bbR^2\), we write \(\partial_i\) for differentiation with respect to the \(i\)-th coordinate. Thus \(\partial_i\fe(h)\) denotes the \(i\)-th component of \(\nabla\fe(h)\), and \(\partial_1\fe(C,s)\) means \(\partial_1\fe\) evaluated at \(h=(C,s)\). For functions introduced with explicitly named scalar variables, we differentiate with respect to those names, writing for example \(\partial_C,\partial_s,\partial_\lambda\), or \(\partial_\mu\). When a Lagrangian depends on a velocity variable \(v=(v_1,v_2)\), we may also write \(\partial_{v_i}\) for differentiation with respect to the \(i\)-th velocity component.

%%%%
\subsection{Self-repelling polymers}
\label{sec:Intro:PathEnsembles}
%%%%

%
\subsubsection{Underlying random walk and path ensembles}
Given a nearest-neighbor path \(\gamma = (\gamma_0, \ldots, \gamma_n)\) in \(\bbZ^2\), we shall use the following observables:
\begin{itemize}
  \item The displacement, \(X(\gamma) \defby \gamma_n-\gamma_0\).
  \item The length, \(\abs{\gamma} \defby n\).
  \item The local time at \(x\in\bbZ^2\), \(\loct_x(\gamma) \defby \sum_{i=0}^n \IF{\gamma_i = x}\).
  \end{itemize}
Given two paths \(\gamma\) and \(\gamma'\), we denote by \(\gamma\concat\gamma'\) the path obtained by concatenating \(\gamma\) and \(\gamma'\), that is, if \(\gamma=(\gamma_0,\gamma_1,\dots,\gamma_n)\) and \(\gamma'=(\gamma'_0,\gamma'_1,\dots,\gamma'_m)\), then \(\gamma\concat\gamma' \defby (\gamma^{\vphantom\prime}_0,\dots,\gamma^{\vphantom\prime}_n,\gamma'_1+(\gamma^{\vphantom\prime}_n-\gamma'_0),\dots,\gamma'_m+(\gamma^{\vphantom\prime}_n-\gamma'_0))\).

\subsubsection{Self-interaction}

Let \(\phi:\bbZ_{\geq 0}\to\bbR\cup\{+\infty\}\) be such that
\begin{equation}\label{eq:condphi}
  \phi(0) = \phi(1) = 0,\qquad\forall m,n\geq 0:\, \phi(m+n) \geq \phi(m) + \phi(n).
\end{equation}
To each path \(\gamma=(\gamma_0,\ldots,\gamma_n)\), we associate the self-interaction energy
\[
  \Phi(\gamma) \defby \sum_{x\in\bbZ^2} \phi(\loct_x(\gamma))
\]
and the weight
\[
  \W(\gamma) \defby e^{-\Phi(\gamma)}.
\]

The superadditivity condition in~\eqref{eq:condphi} ensures that the self-interaction is repulsive: it energetically penalizes multiple occupations of the same site, thereby favoring trajectories in which sub-paths tend to avoid each other. Three well-known examples are the simple random walk (which corresponds to the choice \(\phi \equiv 0\)), the self-avoiding walk (which corresponds to the choice \(\phi(\ell)=+\infty\cdot\IF{\ell>1}\)) and the Domb--Joyce model
(which corresponds to the choice \(\phi(\ell)=\beta\ell(\ell-1)/2\)).

\begin{remark}
The assumption that \(\phi(1)=0\) is just a normalization, and does not lead to a loss of generality: we will be working with polymers of fixed length, and replacing \(\phi(m)\) by \(\phi(m)-m\phi(1)\) preserves superadditivity and only multiplies the weight of every path of length \(n\) by the constant factor \(e^{\phi(1)(n+1)}\). The fixed-length probability measures are therefore unchanged; only the corresponding thermodynamic quantities are shifted by constants.
\end{remark}

\subsubsection{Path ensembles}
We consider probability measures on three natural families of paths: our main interest is in paths of fixed length and fixed displacement, but we shall also use paths of fixed length but arbitrary displacement, or paths of fixed displacement but arbitrary length. We have tried to use consistent notations, putting intensive parameters as superscripts and extensive parameters as subscripts.

\paragraph{Fixed-length, fixed displacement ensemble.}

Given \(n\geq 1\) and \(x\in\bbZ^2\), the fixed-length, fixed-displacement ensemble is defined by the
probability measure on paths starting at \(0\) given by
\[
  \p{}{n,x}(\gamma) \defby \frac1{\Z{}{n,x}}\, \IF{\abs{\gamma}=n,\, X(\gamma)=x} \W(\gamma),
\]
where \(\Z{}{n,x}\defby \sum_\gamma \IF{\abs{\gamma}=n,\, X(\gamma)=x} \W(\gamma)\).

\paragraph{Fixed-length ensemble.}

Given \(n\geq 1\) and \(h\in\bbR^2\), the fixed-length ensemble is defined by the probability measure on paths starting at \(0\) given by
\[
  \p{h}{n}(\gamma) \defby \frac1{\Z{h}{n}}\,e^{\spr{h}{X(\gamma)}} \IF{\abs{\gamma}=n} \W(\gamma),
\]
where \(\Z{h}{n}\defby \sum_\gamma e^{\spr{h}{X(\gamma)}} \IF{\abs{\gamma}=n} \W(\gamma)\), and \(\spr{x}{y}\) denotes the usual inner product in \(\bbR^2\).

\paragraph{Fixed-displacement ensemble.}

Given \(x\in\bbZ^2\) and \(\lambda\in\bbR\), the fixed-dis\-place\-ment ensemble is defined by the probability measure on paths starting at \(0\) given by
\[
  \p{\lambda}{x}(\gamma) \defby \frac1{\Z{\lambda}{x}}\, e^{-\lambda\abs{\gamma}} \IF{X(\gamma)=x} \W(\gamma),
\]
where \(\Z{\lambda}{x}\defby \sum_\gamma e^{-\lambda\abs{\gamma}} \IF{X(\gamma)=x} \W(\gamma)\).
It will be convenient to allow the case in which \(\abs{\gamma}=0\) (i.e., when \(\gamma\) is composed of the single vertex \(0\)). In particular, \(\Z{\lambda}{0} \geq 1\) even for the self-avoiding walk.

Note that, unlike the two previous ensembles, the partition function \(\Z{\lambda}{x}\) is infinite when \(\lambda\) is too small, and the measure \(\p{\lambda}{x}\) is not well defined when this happens. We shall return to this issue below.

\bigskip
In all three cases, we denote by \(\E{}{}\), \(\V{}{}\), with the appropriate superscripts and/or subscripts, the expectation and the variance in the corresponding ensemble: \(\E{\lambda}{x}\), \(\V{h}{n}\), etc.

%%%%
\subsection{Thermodynamic quantities}
\label{sec:Intro:Thermo}
%%%%

In this section, we introduce the basic thermodynamic quantities associated to the partition functions \(\Z{h}{n}\), \(\Z{\lambda}{x}\) and \(\Z{}{n,x}\).
Their main properties are discussed in Section~\ref{sec:Thermo}.
\subsubsection{Free energy}
For any \(h\in\bbR^2\), the \emph{free energy} is defined by
\[
  \fe(h) \defby \lim_{n\to\infty} \fe_n(h),
  \quad\text{where}\quad
  \fe_n(h) \defby \tfrac 1n \log\Z{h}{n}.
\]
It is easy to check, using standard subadditivity arguments, that the limit exists for all \(h\in\bbR^2\) and \(\Z{h}{n} \geq e^{n\fe(h)}\) for all \(n\) and \(h\).
Moreover, convexity of \(\fe\) follows immediately from Hölder's inequality, and rough bounds can be used to establish finiteness of \(\fe\).

\subsubsection{Inverse correlation length}
Let us write \(\lambdac \defby \fe(0)\). The fixed-displacement ensemble is well defined when \(\lambda>\lambdac\). Indeed, \[
  \forall \lambda > \lambdac,\qquad
  \Z{\lambda}{x} \leq \sum_{x\in\bbZ^2} \Z{\lambda}{x} = \sum_{n\geq 0} e^{-\lambda n} \Z{0}{n} = \sum_{n\geq 0} e^{-(\lambda-\lambdac)n + \sfo(n)} < \infty .
\]
On the other hand, the susceptibility diverges when \(\lambda\leq\lambdac\), since \(\Z{0}{n} \geq e^{n\fe(0)} = e^{n\lambdac}\) and thus
\[
  \forall\lambda\leq\lambdac,\qquad
  \sum_{x\in\bbZ^2} \Z{\lambda}{x} = \sum_{n\geq 0} e^{-\lambda n} \Z{0}{n} \geq \sum_{n\geq 0} e^{-(\lambda-\lambdac)n} = \infty.
\]

Given \(x=(x_1,x_2)\in\bbR^2\), we write \([x] \defby (\lfloor x_1\rfloor, \lfloor x_2\rfloor)\in\bbZ^2\).
For any \(\lambda > \lambdac\) and any \(x\in\bbR^2\), the \emph{inverse correlation length} is defined by
\[
  \icl_\lambda(x)
  \defby -\lim_{k\to\infty} \tfrac 1k \log\Z{\lambda}{[kx]}.
\]
Again, subadditivity arguments (see, for instance, \cite[Section~3.A.2]{Ioffe+Velenik-2008}) imply that the limit exists and defines a norm on \(\bbR^2\).

\subsubsection{The rate function}
Let \(v\in\bbR^2\) such that \(\normI{v}\leq 1\). Consider a sequence \((x_n)_{n\geq 1}\) such that \(\normI{x_n}\equiv n\pmod 2\) for all \(n\geq 1\), and \(\lim_{n\to\infty} x_n/n = v\). The \emph{rate function} is defined as the Legendre--Fenchel transform of \(\fe\):
\[
  J(v) \defby \sup_{h\in\bbR^2} \bigl\{\spr{h}{v}-\fe(h)\bigr\}.
\]

By the usual convex-duality argument for the fixed-length ensemble, this function also satisfies
\[
  J(v) = -\lim_{n\to\infty}\frac1n\log \Z{}{n,x_n}.
\]
Since the polymer is nearest-neighbor, the effective domain of \(J\) is the closed unit \(\ell^1\)-ball: \(J(v)<\infty\) if and only if \(\normI v\leq 1\). The logarithmic estimates used later are stated in Section~\ref{sec:logAsymptotics}.

Physically, \(J\) is thus the stretching free energy per monomer associated with a macroscopic displacement vector \(v\). We prefer to refer to it as the rate function, in order to avoid confusion with the free energy \(\fe\).

%%%%
\subsection{Main results}
\label{sec:Intro:MainResults}
%%%%

\begin{figure}[t]
  \includegraphics[width=\textwidth]{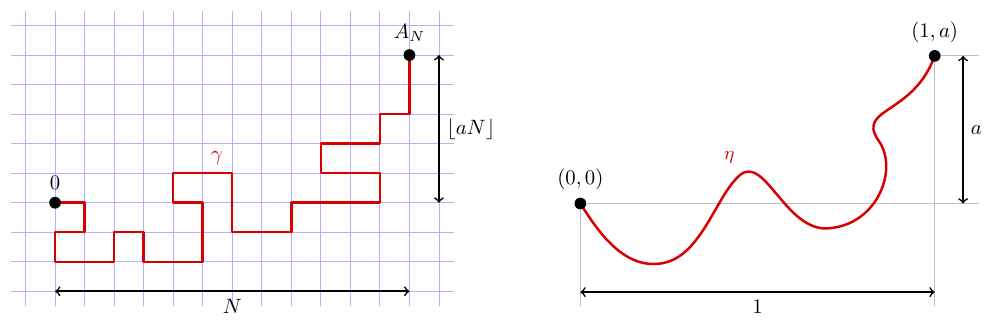}
  \caption{The microscopic setup (left), and the corresponding continuum setup (right).}
  \label{fig:Setup}
\end{figure}

Our central interest in this paper is to analyze the effect of a constant force field on the macroscopic geometry of a polymer of fixed length, with both endpoints fixed at a macroscopic distance from each other.

Let \(g>0\), \(a\in\bbR\) and \(\alpha>1+\abs{a}\). For each \(N\in\bbZ_{>0}\), let \(A_N\defby (N, \lfloor aN \rfloor)\) and \(L_N \in \{\lfloor \alpha N \rfloor, \lfloor \alpha N \rfloor+1\}\) be such that \(L_N\equiv\normI{A_N}\pmod 2\).
We consider the following probability measure on paths \(\gamma\) starting at \(0\) (see Fig.~\ref{fig:Setup}):
\[
  \p{g}{L_N,A_N}(\gamma) \defby \frac{1}{\Z{g}{L_N,A_N}} \IF{X(\gamma)=A_N} \IF{\abs{\gamma}=L_N} \W(\gamma) \exp\Bigl( - \frac gN \sum_{k=1}^{L_N} (\gamma_k)_2 \Bigr),
\]
where
\[
  \Z{g}{L_N,A_N} \defby \sum_{\substack{\gamma: 0 \to A_N \\ \abs{\gamma}=L_N}} \W(\gamma) \exp\Bigl( - \frac gN \sum_{k=1}^{L_N} (\gamma_k)_2 \Bigr).
\]
\begin{remark}
  We do not explicitly introduce an inverse temperature into our notation, as this parameter remains fixed throughout our analysis. This choice does not alter the behavior of either the self-avoiding walk or the simple random walk, since in both cases the self-interaction is unaffected; specifically, \(\Phi = \beta\Phi\) for all \(\beta>0\). In particular, increasing \(\beta\) is precisely equivalent to increasing \(g\). For more general self-interactions, changing \(\beta\) affects the relevant thermodynamic quantities in non-trivial ways. Nevertheless, our results apply to any repulsive self-interaction, and varying \(\beta\) does not alter this property. Consequently, the effect of varying \(\beta\) can be deduced from our results, provided a sufficient understanding of the \(\beta\)-dependence of these quantities is available.
\end{remark}

\subsubsection{Variational problems}
Let \(\alpha>0\), \(g>0\), and \(a\in\bbR\). We denote by \(\mathcal{AC}\) the set of absolutely continuous curves
\[
  \MP:[0,1]\to\bbR^2,
  \qquad
  \MP(0)=(0,0),\quad \MP(1)=(1,a).
\]
We shall consider two variational problems. The first one is written in terms of the inverse correlation length, while the second one is its dual formulation in terms of the rate function \(J\).

The primal functional is
\[
  \calA(\MP,\lambda)
  \defby
  \int_0^1 \icl_{\lambda+g\MP_2(t)}(\dot\MP(t))\,\dd t - \lambda\alpha,
\]
defined on pairs \((\lambda,\MP)\in\bbR\times\mathcal{AC}\) for which \(\mu_{\lambda,\MP}(t)\defby \lambda+g\MP_2(t) > \lambdac\) for a.e.\ \(t\), and with value \(+\infty\) otherwise.

\smallskip\noindent
\textit{\textbf{Primal variational problem:} Minimize \(\calA(\MP,\lambda)\) over all pairs \((\lambda,\MP)\in\bbR\times\mathcal{AC}\) satisfying the constraint
\[
  \int_0^1 \partial_\mu\icl_\mu(\dot\MP(t)) \vert_{\mu=\lambda+g\MP_2(t)} \,\dd t = \alpha.
\]}

\medskip
The dual functional is
\[
  \calA^*(\MP,\rho) \defby
  \int_0^1 \Bigl[ g\rho(t)\MP_2(t) + \rho(t)J\bigl(\dot\MP(t)/\rho(t)\bigr)\Bigr] \dd t,
\]
defined on pairs \(\MP\in\mathcal{AC}\) and \(\rho\in L^1([0,1],\bbR_{\geq 0})\), with the usual convention for the perspective function:
\[
  \rho J(v/\rho)=
  \begin{cases}
    0, & \rho=0,\ v=0,\\
    +\infty, & \rho=0,\ v\ne0.
  \end{cases}
\]

\smallskip\noindent
\textit{\textbf{Dual variational problem:} Minimize \(\calA^*(\MP,\rho)\) under the constraint
\[
  \int_0^1 \rho(t)\,\dd t=\alpha.
\]}

\medskip
A useful way to interpret the dual problem is to introduce the monomer-time variable
\begin{equation}\label{eq:Intro:MonomerTime}
  q(t) \defby \int_0^t \rho(u)\,\dd u.
\end{equation}
Since the dual constraint imposes \(q(1)=\alpha\), this reparameterizes the curve on the interval \([0,\alpha]\). If
\(\eta(q(t))=\MP(t)\), then
\[
  \calA^*(\MP,\rho)
  =
  \int_0^\alpha \bigl[J(\eta'(q))+g\eta_2(q)\bigr]\dd q.
\]
Thus the dual variational problem is equivalently the minimization of the fixed-time functional
\[
  \calB(\eta)
  =
  \int_0^\alpha \bigl[J(\eta'(q))+g\eta_2(q)\bigr]\dd q,
\]
over absolutely continuous curves \(\eta:[0,\alpha]\to\bbR^2\) joining \((0,0)\) to \((1,a)\).

\bigskip
In the sequel, we shall say that \((\lambda,\MP)\) is primal-admissible if \(\MP\in\mathcal{AC}\), \(\calA(\MP,\lambda)<\infty\), and the primal microscopic-length constraint is satisfied.
Similarly, \((\MP,\rho)\) is dual-admissible if \(\MP\in\mathcal{AC}\), \(\rho\in L^1([0,1],\bbR_{\geq0})\), \(\calA^*(\MP,\rho)<\infty\), and \(\int_0^1\rho(t)\,\dd t=\alpha\).
Finally, a curve \(\eta\) is admissible for \(\calB\) if \(\eta\in\mathcal{AC}([0,\alpha];\bbR^2)\), \(\eta(0)=(0,0)\), \(\eta(\alpha)=(1,a)\), and \(\calB(\eta)<\infty\).

\begin{remark}
  The primal functional and the primal constraint are invariant under increasing absolutely continuous reparameterizations of the curve. The dual problem has the same geometric invariance, provided the density \(\rho\) is transformed accordingly. Thus no arbitrary parameterization of the curve is selected by the variational problem.

  The density \(\rho\) should nevertheless be viewed as part of the limiting object: the measure \(\rho(t)\,\dd t\) records the inhomogeneous density of monomers along the macroscopic trace. Equivalently, after passing to monomer time, this information is encoded in the curve \(\eta\).
\end{remark}

The primal and dual variational problems are equivalent in the following sense.
\begin{theorem}\label{thm:Main:EquVP}
  The primal and dual variational problems have the same minimal value.
  Moreover, their minimizers are related as follows.

  If \((\lambda,\MP)\) is a primal minimizer, then
  \[
    \rho(t)
    =
    \partial_\mu\icl_\mu(\dot\MP(t))
    \vert_{\mu=\lambda+g\MP_2(t)}
  \]
  defines a dual minimizer \((\MP,\rho)\).

  Conversely, if \((\MP,\rho)\) is a dual minimizer, then there exists \(\lambda_*\in\bbR\) such that \((\lambda_*,\MP)\) is a primal minimizer.

  Moreover, for corresponding minimizers,
  \[
    \calA(\MP,\lambda)=\calA^*(\MP,\rho).
  \]
\end{theorem}

\begin{remark}
The primal/dual equivalence above is stated at the level of minimizers. This is intentional. Starting from a primal-admissible pair \((\lambda,\MP)\), one can formulate a corresponding dual pair by using the subdifferential form of the perspective identity. Namely, whenever \(\rho(t)\in \partial_\mu \icl_\mu(\dot\MP(t)) \vert_{\mu=\lambda+g\MP_2(t)}\) can be chosen in \(L^1\), the associated pair \((\MP,\rho)\) is dual-admissible and satisfies \(\calA(\MP,\lambda)=\calA^*(\MP,\rho)\).
In the strictly supercritical regime, where \(\lambda+g\MP_2(t)>\lambdac\) uniformly, the subdifferential reduces to the usual derivative and this correspondence is smooth.

The converse direction is more restrictive for arbitrary dual-admissible pairs. A finite value of \(\calA^*(\MP,\rho)\) imposes the kinematic constraint encoded by the perspective function, but it does not by itself ensure that the pair comes from a primal pair. For this, one would need the existence of a single global multiplier \(\lambda\) such that \(\mu(t) = \lambda + g\MP_2(t)\) and such that the perspective representation of \(\icl_{\mu(t)}(\dot\MP(t))\) is saturated for a.e.\ \(t\). This is an additional compatibility condition, not part of dual admissibility.

Thus the dual formulation does not merely allow more density profiles on a fixed trace. Even after projection onto traces, the dual admissible class is in general larger than the class of traces arising from primal-admissible curves. Indeed, dual admissibility only requires the existence of some monomer-density profile satisfying the mass constraint and making the perspective term finite; equivalently, since \(\mathrm{dom}(J)=\setof{v\in\bbR^2}{\normI v\leq 1}\), \(\normI{\dot\MP(t)}\leq \rho(t)\) for a.e.\ \(t\). By contrast, primal admissibility requires the existence of a single global multiplier \(\lambda\) for which the perspective identity is saturated along the curve. The dual formulation is therefore the natural one for arbitrary monomer-density profiles and for the coarse-graining estimates in Section~\ref{sec:ProofMain}, while the primal Finsler formulation describes the saturated pairs and is most useful for identifying the minimizing trajectory.

At the minimizer constructed below, the critical barrier keeps the effective parameter uniformly in the strictly supercritical regime; consequently all quantities are smooth there, and the primal/dual correspondence is unambiguous.
\end{remark}
The two formulations will both be used in the proof: the primal problem is the natural geometric formulation, while the dual problem becomes particularly simple after passing to monomer time. By Theorem~\ref{thm:Main:EquVP}, it is enough to identify the minimizer in one of the two formulations. We state the result in primal variables, which gives the geometric interpretation, but the explicit construction is most naturally expressed in monomer time.
\begin{theorem}\label{thm:SolutionPrimalVP}
  Assume that \(g>0\), \(a\in\bbR\), and \(\alpha>1+\abs a\). The primal variational problem possesses a unique minimizing trace among absolutely continuous curves.

  More precisely, there exist unique constants \(C_*>0\), \(s_*\in\bbR\), and
  \[
    \lambda_*=\fe(C_*,s_*),
  \]
  such that the monomer-time curve
  \[
    \eta_*(q)
    =
    \int_0^q \nabla\fe(C_*,s_*+gu)\,\dd u,
    \qquad q\in[0,\alpha],
  \]
  connects \((0,0)\) to \((1,a)\). Every primal minimizer is of the form
  \[
    \MP(t)=\eta_*(\theta(t))
  \]
  for some nondecreasing absolutely continuous map
  \[
    \theta:[0,1]\to[0,\alpha],
    \qquad
    \theta(0)=0,\quad \theta(1)=\alpha.
  \]
  Conversely, every such reparameterization is a primal minimizer, with
  \(\lambda=\lambda_*\).

  The trace of \(\eta_*\) is the graph of an analytic function \(y_*\).
  In graph parameterization, the minimizer is therefore unique and analytic.
  Finally, there exists \(\delta_*>0\) such that
  \[
    \lambda_*+gy_*(x)\geq \lambdac+\delta_*
    \qquad\text{for all }x\in[0,1].
  \]
\end{theorem}
\begin{remark}
  The primal constraint acts as a critical barrier. Indeed, by point~\ref{lem:InvCorrLength:dmunumu} of Lemma~\ref{lem:InvCorrLength}, \(\partial_\mu\icl_\mu(v)\longrightarrow+\infty\) as \(\mu\downarrow\lambdac\) and \(v\neq 0\). Thus an admissible curve cannot spend a positive amount of moving time at the critical boundary \(\mu=\lambdac\). For arbitrary competitors, this does not by itself imply a uniform positive distance from criticality. For the minimizing curve, however, the stronger conclusion holds: by Theorem~\ref{thm:SolutionPrimalVP}, there exists \(\delta_*>0\) such that \(\lambda_* + gy_*(x) \geq \lambdac + \delta_*\) for all \(x\in[0,1]\).

  The same phenomenon has a dual interpretation. Although the dual admissible class contains no explicit critical height, any pair which saturates the primal/dual relation inherits the same obstruction. In the perspective identity \(\icl_\mu(v) = \inf_{\rho\geq 0} \rho\bigl(\mu+J(v/\rho)\bigr)\), the minimizing density is \(\rho=\partial_\mu\icl_\mu(v)\). Hence, for \(v\neq 0\), approaching the critical boundary \(\mu=\lambdac\) forces the corresponding monomer density to diverge. Thus the critical singularity of the primal formulation appears in the dual formulation as a loss of compactness of the density variable. The constraint \(\int\rho=\alpha\) prevents such a behavior on a set of positive macroscopic measure for near-minimizing saturated pairs.
\end{remark}

\begin{remark}
  Recall that a Finsler metric is a generalization of a Riemannian metric in which the infinitesimal cost of a displacement \(v\) at a point \(x\) is given by a norm \(F(x,v)\), not necessarily induced by a quadratic form; see~\cite{Bao+Chern+Shen-2000} for background. The length of a curve is then
  \[
    \int F(\MP(t),\dot\MP(t))\,\dd t.
  \]
  In our case, for fixed \(\lambda\), the primal action has this form with \(F_\lambda((x_1,x_2),v)=\icl_{\lambda+gx_2}(v)\). Thus the minimizing curve can be viewed as a geodesic for a height-dependent Finsler metric.

  This Finsler interpretation is useful for understanding the geometry of the limiting curve, but it is not the most natural framework for proving global uniqueness. The primal formulation describes the saturated pairs associated with a global multiplier \(\lambda\), whereas the dual problem is posed on the larger class of trace-density pairs \((\MP,\rho)\). Proving uniqueness in the dual formulation therefore yields uniqueness in the primal formulation, while the converse would not in general control all dual competitors.

  For this reason, in Section~\ref{sec:VP} we prove global uniqueness using the convex-analytic structure of the monomer-time functional. The Finsler viewpoint remains useful for interpreting the minimizing trace, while the dual formulation provides the appropriate variational setting for proving uniqueness and stability.
\end{remark}

\begin{remark}
  It is worth emphasizing that the variational problem does not separate into a purely elastic term plus an independent gravitational potential. One might have expected a functional of the form
  \[
    \int \tau(\dot\MP(t))\,\dd t + \int V(\MP_2(t))\,\dd t,
  \]
  as in many geometric variational problems with an external field. This is not what happens here. In the primal formulation, the field enters directly into the local inverse correlation length through the parameter \(\mu_{\lambda,\MP}(t) = \lambda+g\MP_2(t)\), so that the effective local metric itself depends on height. In the dual formulation, the two effects are also coupled through the monomer density \(\rho\):
  \[
    \int_0^1 \bigl[ g\rho(t)\MP_2(t) + \rho(t)J(\dot\MP(t)/\rho(t)) \bigr] \dd t.
  \]
  Thus the gravitational field affects not only the macroscopic position of the curve, but also the local allocation of microscopic length along it. Although, after passing to monomer time, the dual functional becomes
  \[
    \int_0^\alpha [J(\eta'(q))+g\eta_2(q)]\,\dd q,
  \]
  this separation is a feature of the microscopic monomer-time parameterization: the optimization over the local monomer density has already been absorbed into the change of variables.
\end{remark}

\begin{remark}
  When \(a=0\), symmetry of the functional and uniqueness imply that the minimizer is symmetric under reflection through the line \(x=\tfrac12\). One might further expect this symmetric minimizer to be convex. This is true for the simple random walk, and Fig.~\ref{fig:SimSAW} suggests that it might also be true for the self-avoiding walk. However, establishing whether this holds in the anisotropic setting considered here seems to be a nontrivial problem. While it is easy to exclude the presence of local maxima away from the boundary, full convexity requires more. The first-order equations derived in Section~\ref{sec:VP} reduce the convexity question to an explicit inequality involving first- and second-order derivatives of \(\fe\). The validity of this inequality for general self-repelling polymers does not seem to be obvious, although it can be verified for the simple random walk.
\end{remark}

\subsubsection{Statement of the main results}
Our first main result provides an expression for the limiting free energy density
\[
  \Psi(g,a,\alpha) \defby - \lim_{N\to\infty} \frac{1}{N}\log\Z{g}{L_N,A_N}.
\]
\begin{theorem}\label{thm:FreeEnergy}
  The limiting free energy density exists and is given by
  \[
    \Psi(g,a,\alpha) = \calA(\MP_*,\lambda_*) = \calA^*(\MP_*,\rho_*),
  \]
  where the minimizers are those described in Theorem~\ref{thm:SolutionPrimalVP}.
\end{theorem}
The second result studies the typical conformations of the polymer under the measure \(\p{g}{L_N,A_N}\). Since we are not only interested in the macroscopic shape of the polymer, but also in the height-dependent distribution of monomers along the curve, we need to encode the microscopic polymer geometry in a suitable way. For a path \(\gamma=(\gamma_0,\ldots,\gamma_{L_N})\), consider the empirical monomer measure
\[
  \frm_\gamma^N \defby \frac1N\sum_{k=1}^{L_N}\delta_{\gamma_k/N}.
\]
Note that its total mass is \(L_N/N\), which converges to \(\alpha\) as \(N\to\infty\).

The limiting monomer measure associated with the minimizer of the variational problem is the push-forward \(\frm_*\defby (\eta_*)_\sharp\Leb_{[0,\alpha]}\) of the Lebesgue measure on \([0,\alpha]\) by \(\eta_*\), that is,
\[
  \int \varphi\,\dd\frm_* = \int_0^\alpha \varphi(\eta_*(q))\,\dd q = \int_0^1 \varphi(x,y_*(x))\,\rho_*(x)\,\dd x ,
\]
where the last expression uses the graph parameterization.

We metrize weak convergence of finite measures by the bounded-Lipschitz distance
\[
  d_{\mathrm{BL}}(\mu,\nu) \defby \sup_{\normsup{\varphi}\leq 1,\ \mathrm{Lip}(\varphi)\leq 1}
  \abs[\Big]{\int\varphi\,\dd\mu-\int\varphi\,\dd\nu}.
\]

The next result shows that the empirical monomer measure concentrates near the deterministic measure \(\frm_*\).
\begin{theorem}\label{thm:Concentration}
  For every \(\epsilon>0\), there exists \(c_\epsilon>0\) such that, for all \(N\) sufficiently large,
  \[
    \p{g}{L_N,A_N} ( d_{\mathrm{BL}}(\frm_\gamma^N,\frm_*) > \epsilon) \leq e^{-c_\epsilon N}.
  \]
\end{theorem}
In particular, the trace of the polymer concentrates near the unique minimizing trace \(\Gamma_*\) described in Theorem~\ref{thm:SolutionPrimalVP}, and the local density of monomers along the limiting trace is given by \(\rho_*(x)\). To formulate the trace concentration precisely, set
\[
  \forall\epsilon>0,\qquad
  \calT_\epsilon^* \defby \setof{z\in\bbR^2}{d_2(z,\Gamma_*)<\epsilon}.
\]
\begin{corollary}\label{cor:TubeConcentration}
  For every \(\epsilon>0\), there exists \(c_\epsilon>0\) such that, for all \(N\) sufficiently large,
  \[
    \p{g}{L_N,A_N} ( N^{-1}\gamma \not\subset \calT_\epsilon^* ) \leq e^{-c_\epsilon N}.
  \]
\end{corollary}
Figures~\ref{fig:SimSAW} and~\ref{fig:SimSRW} illustrate this convergence of the trace in the case of the self-avoiding walk and the simple random walk, while Figure~\ref{fig:density} illustrates convergence of the monomer density.
\begin{figure}[t]
	\centering
	\includegraphics[width=\textwidth]{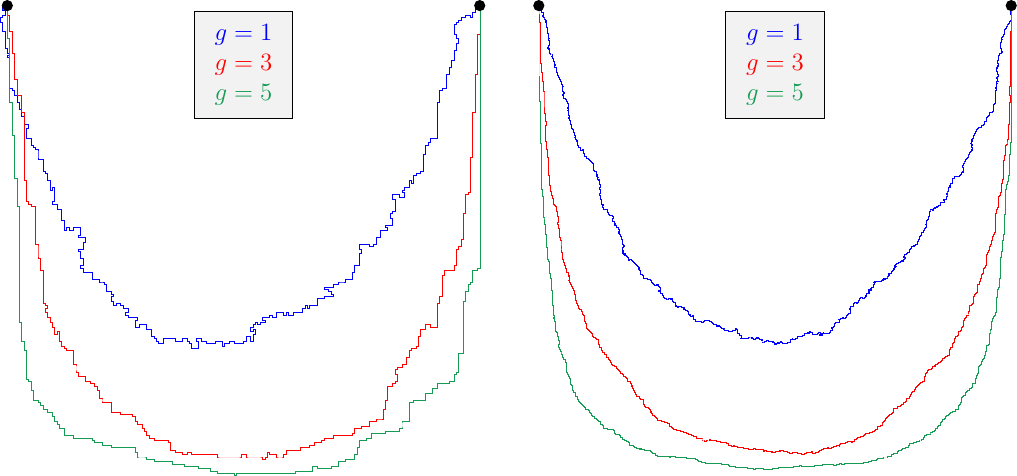}
	\caption{Simulation of the model with \(\alpha=3, a=0\) and \(N=200\) (left) and \(N=1000\) (right) for various values of \(g\), in the case of the self-avoiding walk.}
	\label{fig:SimSAW}
\end{figure}
\begin{figure}[t]
	\centering
	\includegraphics[width=\textwidth]{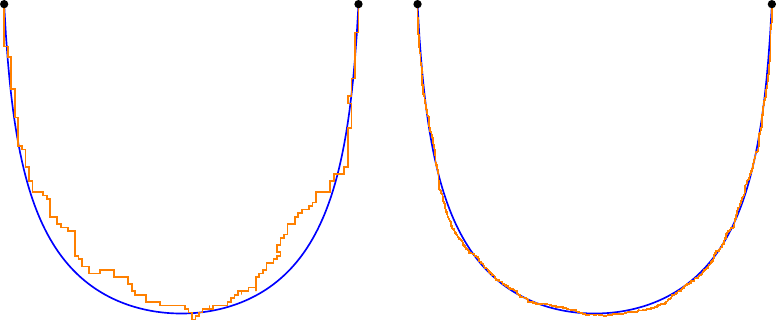}
	\caption{Simulation of the model with \(\alpha=3, a=0, g=3\) and \(N=100\) (left) and \(N=1000\) (right), in the case of the simple random walk. The blue curve is the geodesic of equation~\eqref{eq:geodesicSRW}.}
	\label{fig:SimSRW}
\end{figure}
\begin{figure}[t]
	\centering
	\includegraphics[width=\textwidth]{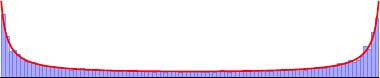}\\
	\includegraphics[width=\textwidth]{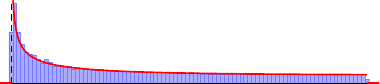}
	\caption{Histogram of the empirical monomer density associated to one trajectory of the simple random walk model with \(\alpha=3\), \(a=0\), \(g=3\) and \(N=10000\). Top: horizontal marginal of the empirical monomer measure (blue) and the corresponding marginal of the variational monomer measure \((\eta_*)_\#\mathrm{Leb}_{[0,\alpha]}\) (red). Bottom: the same for the vertical marginal; the dashed line marks the lower edge of the support of the limiting measure, where the theoretical density has an integrable singularity.}
	\label{fig:density}
\end{figure}

\bigskip
We conclude the introduction by discussing several special regimes in which the variational description becomes more explicit.
\subsubsection{An exactly solvable case}
For general self-interactions, there is little hope of computing the thermodynamic functions \(\fe\), \(\icl_\lambda\) and \(J\) in closed form; in particular, the explicit determination of the geodesics is generally intractable. An exception is the case of the simple random walk, where \(\Phi(\gamma)\equiv 0\). In this case, all relevant thermodynamic quantities and geodesics can be computed explicitly; see Section~\ref{sec:VP:SRW}.

\begin{figure}[t]
	\centering
	\includegraphics[height=6cm, valign=t]{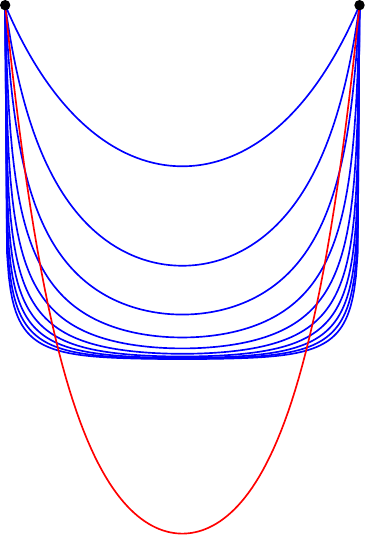}
	\hspace{2cm}
	\includegraphics[width=4cm, valign=t]{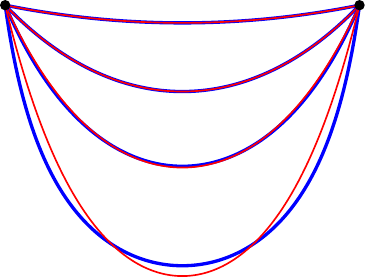}
	\caption{Left: The geodesics of equation~\eqref{eq:geodesicSRW} for \(a=0\), \(\alpha=3\) and values of \(g\) ranging from 1 to 10. The standard catenary for an inextensible string of length 3 is drawn in red for comparison. Right: The geodesics (blue) and the catenaries (red) of same apparent length, for \(a=0\), \(\alpha=3\) and increasing values of \(g\): \(0.1, 0.5, 1, 2\).}
	\label{fig:SimCatenary}
\end{figure}

Since we have an explicit expression for the geodesics (see~\eqref{eq:geodesicSRW}), we can discuss in detail their behavior and compare it to the classical catenary. Figure~\ref{fig:SimCatenary} (left) illustrates the geodesics for the simple random walk with \(a=0\), \(\alpha=3\), and several values of \(g\), together with classical catenaries.

Several qualitative features are immediately visible. First, the apparent macroscopic length of the geodesics is strictly
smaller than \(\alpha\), reflecting the fact that the polymer stores part of its length in microscopic fluctuations rather than in macroscopic extension. Only in the limit \(g\to\infty\) does the polymer reach its maximal macroscopic extension, and the limiting curve becomes piecewise affine, connecting \((0,0)\) to \((0,-1)\), then to \((1,-1)\) and finally to \((1,0)\).

Second, as \(g\) increases, the geodesics deviate significantly from the classical catenary and develop a flat-bottomed profile. This limiting ``boxy'' shape is a direct manifestation of the lattice structure. Indeed, in the limit \(g\to\infty\), the polymer is forced into maximal local extension along the path, corresponding to one monomer per unit microscopic length. In a rotationally invariant setting (for instance a continuous isotropic random walk in \(\mathbb R^2\)), this would lead back to the classical catenary. On the lattice \(\mathbb Z^2\), however, the alignment of increments with the coordinate axes allows the path to switch from a vertical descent to a horizontal plateau without changing its microscopic density, leading to the polygonal profile observed in Figure~\ref{fig:SimCatenary}.

Nevertheless, while the detailed shape at large values of \(g\) depends on the microscopic structure of the model (and in particular on the allowed directions of the increments), a form of universality still holds in this regime, in the sense that all nearest-neighbor models with the same set of admissible directions exhibit the same qualitative limiting geometry. This is made precise in Section~\ref{sec:Intro:Largeg}.

In contrast, the small-field regime exhibits a different, stronger form of universality. As illustrated in Figure~\ref{fig:SimCatenary} (right) for small values of \(g\) the geodesics are remarkably well approximated by classical catenaries, once the length of the latter is chosen to match the apparent length of the polymer. This is discussed in greater generality in the next subsection.

\subsubsection{Comparison with the classical catenary}
As discussed above, in the case of the simple random walk, the geodesics at small values of \(g\) are very closely approximated by the classical catenary, once the length of the catenary is chosen to match the apparent length of the polymer.
The next result shows that this is true of all models considered here, and the agreement actually extends beyond leading order: after matching the apparent length, the two curves remain close to third order in \(g\) in the symmetric case.

\begin{proposition}\label{pro:smallg}
Assume that \(a=0\) and let \(g>0\) be sufficiently small. Let \(y_*:[0,1]\to\bbR\) be the minimizing graph constructed in Section~\ref{sec:VP}, with \(y_*(0)=y_*(1)=0\). Let \(y_{\mathrm{cat}}\) be the classical symmetric catenary connecting \((0,0)\) and \((1,0)\), whose Euclidean length is chosen equal to the apparent length of the curve \(y_*\). Then there exists \(C>0\) such that
\[
  \sup_{x\in[0,1]}
  \abs{y_*(x)-y_{\mathrm{cat}}(x)}
  \leq Cg^3 .
\]
\end{proposition}

\begin{remark}
The restriction to \(a=0\) is not merely technical. In the nonsymmetric case, the first-order correction to the straight segment has the form
\[
  g\,\frac{\alpha^2}{2}
  \left(
    \partial_{22}\fe(h_0)
    -
    a\,\partial_{12}\fe(h_0)
  \right)(x^2-x),
\]
where \(h_0=(C_0,s_0)\) is determined by \(\alpha\nabla\fe(h_0)=(1,a)\).
Thus the sign of the leading bending coefficient is governed by
\[
  \partial_1\fe(h_0)\partial_{22}\fe(h_0)
  -
  \partial_2\fe(h_0)\partial_{12}\fe(h_0),
\]
or equivalently by the monotonicity of
\[
  s\mapsto
  \frac{\partial_2\fe(C_0,s)}
       {\partial_1\fe(C_0,s)}
\]
at \(s=s_0\). This is exactly the local condition that appears when one tries to prove convexity of the geodesic. We do not know whether this condition holds under our general assumptions. Extending the catenary comparison to \(a\neq 0\) therefore remains open
for the same reason as the convexity question discussed above. (For the simple random walk, where the monotonicity can be checked explicitly, the leading coefficient is positive; in that case one does recover the nonsymmetric comparison with the catenary up to order \(g^2\).)
\end{remark}
\begin{remark}
  The error estimate in Proposition~\ref{pro:smallg} is sharp, at least at the level of its order in \(g\). Indeed, in the case of the simple random walk with \(a=0\) and \(\alpha=3\), the height of the geodesic at its minimum can be computed explicitly and compared with the height of the catenary of the same apparent length. One finds
  \[
    \abs{y_{\min}^{\rm SRW}-y_{\min}^{\rm cat}}
    =
    \frac1{160}g^3+\sfo(g^3).
    \qedhere
  \]
\end{remark}

\subsubsection{Large-\texorpdfstring{\(g\)}{g} regime}\label{sec:Intro:Largeg}

At the opposite end of the range of fields, the variational problem also has a universal limiting behavior within the class of nearest-neighbor models considered here. As \(g\to\infty\), the gravitational term dominates, while the microscopic nearest-neighbor constraint remains visible through the effective condition \(\normI{\eta'(q)}\leq 1\). The minimizer therefore converges to the curve which descends as fast as possible, travels horizontally at the lowest height compatible with the endpoint and length constraints, and then ascends to the endpoint.
\begin{proposition}[Large-field limit]\label{pro:largeg}
Fix \(a\in\bbR\) and \(\alpha>1+\abs a\). For each \(g>0\), let \(\eta_g:[0,\alpha]\to\bbR^2\) be the minimizing curve in monomer time given by Theorem~\ref{thm:SolutionPrimalVP}. Set
\[
  b_* \defby \frac{1+a-\alpha}{2},
  \qquad
  q_- \defby -b_*=\frac{\alpha-1-a}{2},
  \qquad
  q_+ \defby q_-+1=\frac{\alpha+1-a}{2}.
\]
Let \(\bar\eta:[0,\alpha]\to\bbR^2\) be the polygonal curve
\[
  \bar\eta(q)
  =
  \begin{cases}
    (0,-q), & 0\le q\leq q_-,\\[2mm]
    (q-q_-,b_*), & q_-\leq q\leq q_+,\\[2mm]
    (1,b_*+q-q_+), & q_+\leq q\leq \alpha.
  \end{cases}
\]
Then
\[
  \lim_{g\to\infty}
  \sup_{q\in[0,\alpha]}
  \normII{\eta_g(q)-\bar\eta(q)}
  =
  0.
\]
In particular, the minimizing traces converge, in Hausdorff distance, to the polygonal trace of \(\bar\eta\).
\end{proposition}

%%%%
\subsection{Roadmap to the paper}
\label{sec:Intro:Roadmap}
%%%%

The remainder of the paper contains the proofs of our results.
The analysis of the thermodynamic quantities \(\fe\), \(\icl_\lambda\) and \(J\) is done in Section~\ref{sec:Thermo}.
We then discuss logarithmic asymptotics for partition functions in Section~\ref{sec:logAsymptotics}.
Section~\ref{sec:VP} is devoted to a detailed analysis of the variational problems.
Finally, the proofs of our main results are provided in Section~\ref{sec:ProofMain}.

%%%%%%%%%%%%%%%%%%%%%%%
\section{Properties of the thermodynamic quantities}
\label{sec:Thermo}
%%%%%%%%%%%%%%%%%%%%%%%

In this section we state and prove (when not available in the literature) many fundamental properties of the central thermodynamic quantities: the free energy, the inverse correlation length and the rate function. Many of these properties play an essential role in our derivations. The most important output for the variational analysis is Lemma~\ref{lem:LevelSets}: the level sets of \(\fe\) form a smooth strictly convex foliation of \(\bbR^2\setminus\{0\}\), and their Gauss maps are analytic diffeomorphisms. This is the geometric input used in Section~\ref{sec:VP}.

%%%%
\subsection{Free energy}
\label{sec:Thermo:FreeEnergy}
%%%%

A fundamental quantity associated to \(\fe(h)\) is \(\bar v_h\defby\nabla\fe(h)\). A priori, convexity of \(\fe\) only
guarantees the existence of this gradient almost everywhere. The next lemma shows that \(\fe\) is analytic away from \(0\), and that \(\nabla\fe(h)\to0\) as \(h\to0\); in particular, \(\fe\) is differentiable everywhere. The vector \(\bar v_h\) corresponds to the macroscopic extension of the polymer. Indeed, the standard convergence of gradients for convex functions gives
\begin{equation}\label{eq:vAsLimit}
  \bar v_h = \lim_{n\to\infty} \frac{1}{n} \nabla \log\Z{h}{n} = \lim_{n\to\infty} \E{h}{n}[X(\gamma)/n].
\end{equation}
In particular, \(\normI{\bar v_h}\leq 1\).

\begin{lemma}\label{lem:FreeEnergy}
  \begin{enumerate}
    \item For all \(h\neq 0\), \(\fe\) is analytic in a neighborhood of \(h\).
    \item For all \(h\neq 0\), \(\bar v_h \neq 0\).
    \item For all \(h\neq 0\), \(\Hess\fe(h)\) is positive definite, where \(\Hess\fe(h)\) denotes the Hessian of \(\fe\) at \(h\).
    \item\label{lem:FreeEnergy:Limit} \(\lim_{h \to 0} \normI{\bar v_h} = 0\).
    \item\label{lem:FreeEnergy:Largeh} \(\lim_{\normII{h}\to\infty} \normI{\bar v_h} = 1\).
  \end{enumerate}
\end{lemma}
\begin{proof}
 	The first three claims are proved in~\cite[Section~3.4.1]{Ioffe+Velenik-2008}.

    \smallskip
    \emph{Proof of 4.} It is proved in~\cite[Chapter~3]{Smirnova-2018} that, for any \(\epsilon>0\), there exist \(C\) and \(c_\epsilon>0\) such that
	\[
      \forall n\geq 1,\qquad
      \p{0}{n}(\normI{X(\gamma)} > \epsilon n) \leq C e^{-c_\epsilon n}.
	\]
	Consider the event \(A_\epsilon \defby \{ \normI{X(\gamma)} > \epsilon n \}\). On the one hand, using the Cauchy--Schwarz inequality,
	\[
      \E{0}{n} [ e^{\spr{h}{X(\gamma)}} \boldsymbol{1}_{A_\epsilon} ]
      \leq \E{0}{n} [ e^{2\spr{h}{X(\gamma)}}]^{1/2}\, \p{0}{n}(A_\epsilon)^{1/2}
      \leq C^{1/2} \, e^{(\normII{h} - c_\epsilon/2) n},
    \]
    since \(\normII{X(\gamma)} \leq n\) when \(\abs{\gamma}=n\).
    On the other hand, by Jensen's inequality,
    \[
      \E{0}{n} [ e^{\spr{h}{X(\gamma)}} ] \geq e^{\E{0}{n} [\spr{h}{X(\gamma)}]} = 1,
    \]
    since \(\E{0}{n}[X(\gamma)]=0\) by symmetry.
    Combining these two bounds, we get
    \[
      \p{h}{n}(\normI{X(\gamma)} > \epsilon n)
      \leq C^{1/2}\, e^{-n (c_\epsilon/2 - \normII{h})}
      \leq C'\, e^{-c'_\epsilon n},
    \]
    for some \(C', c'>0\) and all \(h\in\bbR^2\) such that \(\normII{h}\leq c_\epsilon/4\). The desired claim now follows from~\eqref{eq:vAsLimit}, since for all such \(h\neq 0\) we have
    \[
      \normI{\bar{v}_h} = \normI{\nabla \fe(h)} \leq \lim_{n\to\infty} \E{h}{n} [ \normI{X(\gamma)} / n ] \leq \lim_{n\to\infty} \bigl(\epsilon + \p{h}{n}(\normI{X(\gamma)} > \epsilon n)\bigr) = \epsilon.
    \]
    Since \(\epsilon>0\) is arbitrary, this proves the claim.

    \smallskip
	\emph{Proof of 5.}
	Note that, for any \(h\in\bbR^2\), \(\max_{x:\,\normI{x}\leq 1} \spr{h}{x} = \normsup{h}\).
    In particular,
    \[
      \fe_n(h) = \frac{1}{n}\log\sum_{\gamma:\,\abs{\gamma}=n} \exp(\spr{h}{X(\gamma)}) \W(\gamma) \geq \spr{h}{\tfrac1n X(\gamma^*)} = \normsup{h},
    \]
    where we have restricted the sum to a single path \(\gamma^*\) satisfying \(\normI{X(\gamma^*)}=n\) and \(\spr{h}{X(\gamma^*)}=n\normsup{h}\); note that \(\W(\gamma^*)=1\), since this path visits each site only once.
    Therefore,
    \begin{equation}\label{eq:Boundsf}
      \forall h\in\bbR^2,\qquad
      \fe(h) \geq \normsup{h}.
    \end{equation}

    Since \(\normI{\bar v_h}\leq 1\), we only need to prove a lower bound. Let \(w \defby h/\normsup{h}\). Since \(\fe\) is convex,
    \[
      \spr{\nabla \fe(h)}{h} \geq \fe(h) - \fe(0).
    \]
    Using~\eqref{eq:Boundsf} and dividing by \(\normsup{h}\), this yields
    \[
      \spr{\nabla \fe(h)}{w} \geq 1 - \frac{\fe(0)}{\normsup{h}}.
    \]
    Since \(\spr{\nabla \fe(h)}{w} \leq \normI{\nabla \fe(h)} \normsup{w} = \normI{\nabla \fe(h)}\),
    \[
      \normI{\nabla \fe(h)} \geq 1 - \frac{\fe(0)}{\normsup{h}}.
    \]
    Consequently,
    \[
      \liminf_{\normII{h} \to \infty} \normI{\nabla \fe(h)} \geq 1.\qedhere
    \]
\end{proof}

%%%%
\subsection{Inverse correlation length}
\label{sec:Thermo:ICL}
%%%%

We first recall a well-known bound (e.g., \cite{Ioffe+Velenik-2008}): for all \(x\in\bbZ^2\), \(\Z{\lambda}{x} \leq \Bl \Z{\lambda}{0} e^{-\icl_\lambda(x)}\), where we have introduced the \emph{bubble diagram} \(\Bl \defby \sum_{x\in\bbZ^2} (\Z{\lambda}{x})^2\). Since \(\sum_x \Z{\lambda}{x}<\infty\) for \(\lambda>\lambdac\), the bubble diagram is finite as well:
\[
  \sum_x (\Z{\lambda}{x})^2
  \leq
  \bigl(\sup_x \Z{\lambda}{x}\bigr) \sum_x \Z{\lambda}{x}
  < \infty.
\]
Let us now turn to various important properties of \(\icl_\lambda\).
\begin{lemma}\label{lem:InvCorrLength}
	The following properties hold for all \(\lambda>\lambdac\):
	\begin{enumerate}
        \item There exists a universal constant \(c>0\) such that, for all \(\lambda>\lambdac\),
        \[
          \max_{\hat x\in\bbS^1}\icl_\lambda(\hat x) \leq c \min_{\hat x\in\bbS^1}\icl_\lambda(\hat x).
        \]
		\item\label{lem:InvCorrLength:SharpTriangleIneq} There exists \(c_\lambda > 0\) such that
		\[
          \forall x,y\in\bbR^2,\qquad
          \icl_\lambda(x) + \icl_\lambda(y) - \icl_\lambda(x+y)
          \geq c_\lambda (\normII{x} + \normII{y} - \normII{x+y}).
		\]
		\item \(x\mapsto \icl_\lambda(x)\) is real analytic on \(\bbR^2\setminus\{0\}\).
		\item\label{lem:InvCorrLength:PropAsFctOfLambda} For all \(x\in\bbS^1\), \(\lambda \mapsto \icl_\lambda(x)\) is real analytic, (strictly) increasing and (strictly) concave.
		\item For all \(x\in\bbS^1\), \(\lim_{\lambda\downarrow\lambdac} \icl_\lambda(x) = 0\).
		\item\label{lem:InvCorrLength:dmunumu} For all \(x\in\bbS^1\), \(\lim_{\lambda\downarrow\lambdac} \partial_\lambda\icl_\lambda(x) = +\infty\).
  \end{enumerate}
\end{lemma}
\begin{proof}
  \begin{enumerate}
  	\item This follows from the invariance of \(\icl_\lambda\) under the symmetries of the square. Indeed, \(\icl_\lambda(x)\leq \icl_\lambda(e_1)\normI{x}\) and \(\icl_\lambda(x)\geq \icl_\lambda(e_1)\normsup{x}\), so that the claim follows by comparing \(\normI{\cdot}\) and \(\normsup{\cdot}\) on \(\bbS^1\).
  	\item The proof is postponed to Section~\ref{sec:Thermo:DelayedProofs}, as it relies on a result established in Lemma~\ref{lem:LevelSets}.
  	\item This follows from Lemma~3.4 in~\cite{Ioffe+Velenik-2008}.
  	\item The proof of real analyticity is postponed to Section~\ref{sec:Thermo:DelayedProofs}, since it relies on a result established in Lemma~\ref{lem:LevelSets}.

  	\smallskip
  	It remains to prove the strict monotonicity and strict concavity. For the former, it suffices to observe that there exists \(c>0\) such that, for any \(x\in\bbS^1\) and any \(k\in\bbZ_{>0}\) large enough,
  	\[
      -\frac{\partial}{\partial\lambda} \frac{1}{k} \log \Z{\lambda}{[kx]} = \frac{1}{k}\E{\lambda}{[kx]} [\abs{\gamma}] \geq \frac{\normI{[kx]}}{k} \geq c > 0.
  	\]
  	Therefore, for any \(\lambda_1<\lambda_2\) and any \(x\in\bbS^1\),
  	\[
      -\frac{1}{k} \log \Z{\lambda_2}{[kx]} + \frac{1}{k} \log \Z{\lambda_1}{[kx]} \geq c(\lambda_2-\lambda_1),
  	\]
  	which yields, for any \(x\in\bbS^1\), \(\icl_{\lambda_2}(x) \geq \icl_{\lambda_1}(x) + c(\lambda_2-\lambda_1)\).

  	The argument for strict concavity is similar. For fixed \(\lambda>\lambdac\) and \(x\in\bbS^1\), the Ornstein--Zernike analysis in~\cite{Ioffe+Velenik-2008} implies that, for \(k\) large enough,
    \[
      \frac1k \V{\lambda}{[kx]}[\abs{\gamma}] \geq c_{\lambda,x} > 0,
    \]
    uniformly in \(\lambda\) in a compact subset of \((\lambdac,\infty)\).
    Therefore
    \[
      \frac{\partial^2}{\partial\lambda^2}
      \left[-\frac1k\log\Z{\lambda}{[kx]}\right]
      =
      -\frac1k\V{\lambda}{[kx]}[\abs{\gamma}]
      \leq -c_{\lambda,x} < 0,
    \]
    and strict concavity follows in the limit.
  	\item This is proved in~\cite[Proposition~A.1]{Ioffe+Velenik-2010}.
  	\item We postpone the proof to Section~\ref{sec:Thermo:DelayedProofs}, as it relies on some results established in Lemma~\ref{lem:LevelSets} below.
  	\qedhere
  \end{enumerate}
\end{proof}

%%%%
\subsection{Sublevel sets}
\label{sec:Thermo:LevelSets}
%%%%

The sublevel sets of the free energy
\[
  \Kl \defby \setof{h\in\bbR^2}{\fe(h) \leq \lambda}
\]
play an essential role in the analysis of these systems.
They are clearly compact convex sets with nonempty interior for all \(\lambda>\lambdac\) (see Fig.~\ref{fig:LevelLines}).
These sets can also be expressed in terms of the inverse correlation length.
\begin{lemma}\label{lem:WulffShape}
  For all \(\lambda>\lambdac\), \(\Kl = \setof{h\in\bbR^2}{\forall x\in\bbR^2,\, \spr{h}{x}\leq \icl_\lambda(x)}\).
\end{lemma}
\begin{proof}
  Let \(\Kl' \defby \setof{h\in\bbR^2}{\forall x\in\bbR^2,\, \spr{h}{x}\leq \icl_\lambda(x)}\).
  Let us first assume that \(h\in\mathring{\K}_\lambda'\). This entails that the number \(\epsilon \defby \min \setof{\icl_\lambda(x) - \spr{h}{x}}{\normII{x}=1}\) is strictly positive. Therefore,
  \[
    \sum_{n\geq 1} e^{-\lambda n} \Z{h}{n}
    = \sum_{x\in\bbZ^2} e^{\spr{h}{x}} \Z{\lambda}{x}
    \leq b_\lambda \Z{\lambda}{0} \sum_{x\in\bbZ^2} e^{-\epsilon \normII{x}}
    < \infty .
  \]
  Since \(\Z{h}{n} \geq  e^{\fe(h)n}\), we conclude that \(\fe(h) < \lambda\). This proves that \(\mathring{\K}_\lambda'\subset \mathring{\K}_\lambda^{\phantom\prime}\). The set \(\Kl'\) is a compact convex body with nonempty interior, hence it is the closure of its interior. Since both sets are closed, we obtain \(\Kl'\subset \Kl^{\phantom\prime}\).

  Let us now assume that \(h\notin\Kl'\). Then there exist a neighborhood \(V\) of \(h\) and \(y\in\bbZ^2\) such that \(\spr{h'}{y} > \icl_\lambda(y)\) for all \(h'\in V\) (we can choose \(y\in\bbZ^2\) by homogeneity and density of rational directions). Therefore, since \(\Z{\lambda}{ry} = \exp\{-r\icl_\lambda(y)+\sfo(r)\}\) by definition of \(\icl_\lambda\),
  \[
    \sum_{n\geq 1} e^{-\lambda n} \Z{h'}{n}
    = \sum_{x\in\bbZ^2} e^{\spr{h'}{x}} \Z{\lambda}{x}
    \geq \sum_{r\geq 1} e^{r(\spr{h'}{y} - \icl_\lambda(y)) + \sfo(r)}
    = +\infty ,
  \]
  which means that \(\fe(h') \geq \lambda\) for all \(h'\in V\).
  This implies that \(\fe(h) > \lambda\). Indeed, suppose that \(\fe(h)\leq \lambda\). Using \(\fe(0)=\lambdac<\lambda\) and convexity of \(\fe\), this would imply that \(\fe((1-\theta)h)<\lambda\) for all sufficiently small \(\theta>0\), contradicting the fact that \(\fe\geq\lambda\) in a neighborhood of \(h\). We conclude that \(\Kl \subset \Kl'\).
\end{proof}

\begin{figure}[t]
	\centering
	\includegraphics{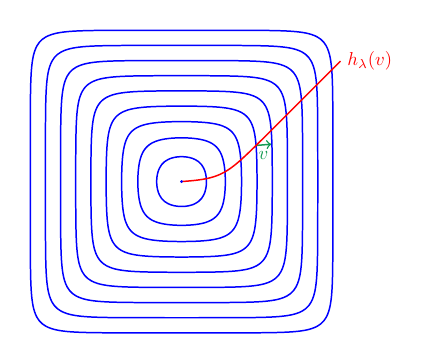}
	\caption{Some level lines of \(\fe\) and the path \(\lambda\mapsto h_\lambda(v)\) of Lemma~\ref{lem:LevelSets} in the case of the simple random walk (that is, \(\Phi(\gamma)\equiv 0\)) and \(v=(\cos(\pi/40),\sin(\pi/40))\).}
	\label{fig:LevelLines}
\end{figure}

Let us state some properties of the sublevel sets.
\begin{lemma}\label{lem:LevelSets}
  \begin{enumerate}
    \item \(\K_{\lambdac} = \{0\}\).
    \item For all \(\lambda_1>\lambda_2\geq\lambdac\), \(\partial\K_{\lambda_1} \cap \partial\K_{\lambda_2}=\emptyset\).
    \item\label{lem:LevelSets:KAnalyticStrictlyConvex} For all \(\lambda>\lambdac\), \(\partial\Kl\) is locally analytic and has strictly positive curvature.
    \item\label{lem:LevelSets:Bijection} The mapping \(h\mapsto \Phi(h) \defby \bigl( \fe(h), \frac{\nabla \fe(h)}{\normII{\nabla \fe(h)}} \bigr)\) defines a bijection between \(\bbR^2\setminus \{0\}\) and \((\lambdac, \infty) \times \bbS^1\).
    \item\label{lem:LevelSets:AnalyticityPaths} Let \(v\in\bbS^1\). For each \(\lambda>\lambdac\), let \(h_\lambda(v)\in\partial\Kl\) be the unique point such that \(\Phi(h_\lambda(v)) = (\lambda, v)\). Then the map \(\lambda \mapsto h_\lambda(v)\) defines an analytic path in \(\bbR^2\). (See Fig.~\ref{fig:LevelLines}.)
    \item\label{lem:LevelSets:Monotonicity} Let \(v\in\bbS^1\). The function \(\lambda \mapsto \normII{\nabla \fe(h_\lambda(v))}\) is strictly increasing on \((\lambdac,\infty)\).
  \end{enumerate}
\end{lemma}

In the sequel, when \(v\in\bbR^2\setminus\{0\}\) is not necessarily a unit vector, we shall use the same notation and write \(h_\lambda(v) \defby h_\lambda(v/\normII{v})\). Thus \(h_\lambda(v)\) denotes the unique point of \(\partial K_\lambda\) whose outward normal is in the direction of \(v\).
\begin{proof}
   \begin{enumerate}
    \item By symmetry, \(0\) is a minimizer of \(\fe\). If \(h\neq 0\) also satisfied \(\fe(h)=\fe(0)\), then by symmetry \(\fe(-h)=\fe(h)\), and convexity would force \(\fe\) to be constant on the segment \([-h,h]\). This would contradict the positive definiteness of the Hessian away from the origin stated in Lemma~\ref{lem:FreeEnergy}. Hence \(0\) is the unique minimizer.
    \item When \(\lambda>\lambdac\), continuity and convexity imply \(\partial \Kl = \setof{h\in\bbR^2}{\fe(h)=\lambda}\). The claim is therefore immediate.
    \item Since \(\fe\) is analytic and \(\nabla\fe(h)\neq 0\) for \(h\neq 0\), the level set \(\{\fe=\lambda\}\) is locally analytic. Its curvature is strictly positive by the positive definiteness of \(\Hess\fe(h)\) on \(\partial\Kl\).
    \item Given \(h\in\bbR^2\setminus\{0\}\), there is a unique \(\lambda>\lambdac\) such that \(h\in\partial\Kl\). Since the curvature of \(\partial \Kl\) is strictly positive, its Gauss map is an analytic diffeomorphism from \(\partial \Kl\) to \(\bbS^1\)~\cite[Section~2.5]{Schneider-1993}.
    \item Observe that \(h_\lambda(v)\) is the unique \(h\in\bbR^2\setminus\{0\}\) satisfying \(\fe(h)=\lambda\) and \(\nabla \fe(h) = \kappa v\) for some \(\kappa>0\). Of course, \(\nabla \fe(h)\) is colinear with \(v\) if and only if \(-v_2 \partial_1 \fe(h) + v_1 \partial_2 \fe(h) = 0\). Therefore, the function \(F:(\bbR^2\setminus\{0\})\times(\lambdac,\infty)\to\bbR^2\) defined by
    \[
      F(h,\lambda) \defby
      \begin{pmatrix}
        \fe(h) - \lambda \\
        -v_2 \partial_1 \fe(h) + v_1 \partial_2 \fe(h)
      \end{pmatrix}
    \]
    is real-analytic and satisfies \(F(h_\lambda(v),\lambda) = 0\). By the analytic implicit function theorem, to ensure that \(\lambda \mapsto h_\lambda(v)\) is analytic in a neighborhood of \(\lambda\), it suffices to check that the Jacobian of \(F\) with respect to \(h\),
    \[
      J_h F =
      \begin{pmatrix}
        \partial_1 \fe & \partial_2 \fe \\
        -v_2 \partial_{11} \fe + v_1 \partial_{21} \fe & -v_2 \partial_{12} \fe + v_1 \partial_{22} \fe
      \end{pmatrix},
    \]
    is non-singular at \(h_\lambda(v)\).
    Substituting the identity \(\nabla \fe(h_\lambda(v)) = \kappa v\) into the first row of the Jacobian determinant, we obtain
    \begin{align*}
      \det(J_h F)
      &= \kappa v_1(-v_2 \partial_{12} \fe + v_1 \partial_{22} \fe) - \kappa v_2(-v_2 \partial_{11} \fe + v_1 \partial_{21} \fe) \\
      &= \kappa \bigl( v_1^2 \partial_{22} \fe - 2v_1 v_2 \partial_{12} \fe + v_2^2 \partial_{11} \fe \bigr).
    \end{align*}
    The second factor in the right-hand side is exactly the quadratic form \((v^\perp)^T \Hess\fe(h_\lambda(v)) v^\perp\), where we have set \(v^\perp\defby (-v_2, v_1)\). Since \(h_\lambda(v)\neq 0\), Lemma~\ref{lem:FreeEnergy} implies that \(\Hess\fe(h_\lambda(v))\) is positive definite, which entails \(\det(J_h F)>0\).
    \item Fix \(v\in\bbS^1\) and set \(\kappa(\lambda) \defby \normII{\nabla \fe(h_\lambda(v))}\). Our goal is to establish that \(\kappa'(\lambda) > 0\) for all \(\lambda>\lambdac\).
    It follows from the definition of the path in the previous point that
    \begin{equation}\label{eq:2Identities}
      \forall \lambda>\lambdac,\qquad
      \fe(h_\lambda(v)) = \lambda
      \qquad\text{and}\qquad
      \nabla \fe(h_\lambda(v)) = \kappa(\lambda)v.
    \end{equation}
    Differentiating the first identity with respect to \(\lambda\) and using the second identity, yields
    \[
      1 = \spr[\big]{\nabla \fe(h_\lambda(v))}{\partial_\lambda h_\lambda(v)}
      = \kappa(\lambda)\, \spr[\big]{v}{\partial_\lambda h_\lambda(v)},
    \]
    and thus
    \begin{equation}\label{eq:condition}
      \spr[\big]{v}{\partial_\lambda h_\lambda(v)} = \frac{1}{\kappa(\lambda)}.
    \end{equation}
    Differentiating the second identity in~\eqref{eq:2Identities} with respect to \(\lambda\), we obtain
    \[
      \Hess\fe(h_\lambda(v))\, \partial_\lambda h_\lambda(v) = \kappa'(\lambda) v.
    \]
    As shown above, \(\Hess\fe(h_\lambda(v))\) is positive definite, which implies that it has a well-defined positive definite inverse \(\Hess\fe^{-1} = \bigl(\Hess\fe(h_\lambda(v))\bigr)^{-1}\).
    Therefore,
    \[
      \partial_\lambda h_\lambda(v) = \kappa'(\lambda) \Hess\fe^{-1} v.
    \]
    Now, substituting this into~\eqref{eq:condition}, we get
    \[
      \kappa'(\lambda) \spr{v}{\Hess\fe^{-1} v} = \frac{1}{\kappa(\lambda)},
    \]
    from which we conclude that
    \[
      \kappa'(\lambda) = \frac{1}{\kappa(\lambda)\, \spr{v}{\Hess\fe^{-1} v}} > 0,
    \]
    since \(\kappa(\lambda) > 0\), \(v\neq 0\), and \(\Hess\fe^{-1}\) is positive definite.
  \end{enumerate}
\end{proof}

%%%%
\subsection{Delayed proofs}
\label{sec:Thermo:DelayedProofs}
%%%%

\begin{proof}[Proof of point~\ref{lem:InvCorrLength:SharpTriangleIneq} of Lemma~\ref{lem:InvCorrLength}]
  By~\cite[Proposition~B.3.1]{Velenik-1997} or~\cite[(208)]{Ioffe-2015}, the inequality follows from the strict convexity of \(\partial\Kl\) (point~\ref{lem:LevelSets:KAnalyticStrictlyConvex} of Lemma~\ref{lem:LevelSets}).
\end{proof}
\begin{proof}[End of proof of point~\ref{lem:InvCorrLength:PropAsFctOfLambda} of Lemma~\ref{lem:InvCorrLength}]
  Fix \(v\in\bbS^1\). It follows from point~\ref{lem:LevelSets:AnalyticityPaths} of Lemma~\ref{lem:LevelSets} that the function \(\lambda\mapsto h_\lambda(v)\) is analytic for \(\lambda>\lambdac\). Since \(\icl_\lambda(v) = \spr{v}{h_\lambda(v)}\), it follows that the function \(\lambda\mapsto \icl_\lambda(v)\) is also analytic for \(\lambda>\lambdac\).
\end{proof}

\begin{proof}[Proof of point~\ref{lem:InvCorrLength:dmunumu} of Lemma~\ref{lem:InvCorrLength}]
  As \(\lambda\downarrow\lambdac\), the points \(h_\lambda(x)\in\partial\Kl\) converge to \(0\), since the sets \(\Kl\) decrease to \(\K_{\lambdac}=\{0\}\). The claim thus follows from point~\ref{lem:FreeEnergy:Limit} of Lemma~\ref{lem:FreeEnergy}, since, by~\eqref{eq:condition},
  \[
    \partial_\lambda\icl_\lambda(x) = \spr{\partial_\lambda h_\lambda(x)}{x} = \frac{1}{\normII{\nabla \fe(h_\lambda(x))}}.
    \qedhere
  \]
\end{proof}

%%%%
\subsection{The rate function}
\label{sec:Thermo:RateFunction}
%%%%

The rate function \(J\) is expressed as the Legendre--Fenchel transform
\begin{equation}\label{eq:JdeV}
  J(v) = \sup_{h\in\bbR^2} \bigl\{\spr{h}{v}-\fe(h)\bigr\} = \sup_{\lambda>\lambdac} \bigl\{\icl_\lambda(v)-\lambda\bigr\},
\end{equation}
where the last identity follows from
\[
  \sup_{h\in\bbR^2}
  \bigl\{\spr{h}{v}-\fe(h)\bigr\}
  =
  \sup_{\lambda>\lambdac}
  \sup_{h\in\partial\K_\lambda}
  \bigl\{\spr{h}{v}-\lambda\bigr\}
  =
  \sup_{\lambda>\lambdac}
  \bigl\{\icl_\lambda(v)-\lambda\bigr\}.
\]

The effective domain of \(J\) is the closed unit \(\ell^1\)-ball. Indeed, the lower bound \(\fe(h)\ge \normsup h\) implies that \(J(v)<\infty\) whenever \(\normI v\leq 1\). Conversely, since \(\W(\gamma)\leq 1\), one has \(\fe(h)\le \log 4+\normsup h\); if \(\normI v>1\), choosing \(h\) in a dual direction to \(v\) then gives \(J(v)=+\infty\).

When \(\normI v<1\), the supremum in~\eqref{eq:JdeV} is attained at the unique \(h\in\bbR^2\) such that \(\nabla\fe(h)=v\). If \(v\neq 0\), then \(h\neq 0\), and
\begin{equation}\label{eq:JdeV_nosup}
  J(v)=\spr{h}{v}-\fe(h).
\end{equation}
For \(v=0\), one has \(J(0)=-\lambdac\). On the boundary \(\normI v=1\), the supremum in~\eqref{eq:JdeV} may fail to be attained at a finite value of \(h\), but its value is still determined by~\eqref{eq:JdeV}; equivalently, it is the lower semicontinuous extension of the interior formula.

Our analysis of the variational problem in Section~\ref{sec:VP} relies on the following elementary identity expressing \(\icl_\mu\) in terms of the perspective function associated with \(J\); see, for instance,~\cite{Rockafellar-1970,Boyd+Vandenberghe-2004}.
\begin{lemma}\label{lem:perspective}
Let \(\mu>\lambdac\). For every \(v\in\bbR^2\),
\[
  \icl_\mu(v) = \min_{\rho\geq 0} \rho\bigl(\mu+J(v/\rho)\bigr),
\]
where the term \(\rho J(v/\rho)\) is interpreted as the perspective of \(J\), namely with the convention
\[
  \rho J(v/\rho)=
  \begin{cases}
    0, & \rho=0,\ v=0,\\
    +\infty, & \rho=0,\ v\ne0.
  \end{cases}
\]
If \(v\neq 0\), the minimum is uniquely achieved at \(\rho^* = \partial_\mu\icl_\mu(v)\).\\
If \(v=0\), the minimum is uniquely achieved at \(\rho=0\), and \(\partial_\mu\icl_\mu(0)=0\).
\end{lemma}
\begin{proof}
  Let first \(\rho>0\). By~\eqref{eq:JdeV},
  \[
    \rho\bigl(\mu+J(v/\rho)\bigr)
    =
    \sup_{\lambda'>\lambdac}
    \bigl\{
      \icl_{\lambda'}(v)
      +
      \rho(\mu-\lambda')
    \bigr\}.
  \]
  Taking \(\lambda'=\mu\) gives \(\rho\bigl(\mu+J(v/\rho)\bigr) \geq \icl_\mu(v)\).
  For \(\rho=0\), the same inequality is immediate from the convention above.

  We now prove that equality can be achieved. If \(v=0\), then \(\icl_\mu(0)=0\), and the value \(\rho=0\) gives equality.

  Assume now that \(v\neq 0\). Let \(h_\mu(v)\in\partial \K_\mu\) be the unique supporting point associated with the direction \(v\). Thus \(\fe(h_\mu(v))=\mu\) and \(\icl_\mu(v)=\spr{h_\mu(v)}{v}\). Since \(\nabla\fe(h_\mu(v))\) is a positive multiple of \(v\), there exists a unique \(\bar\rho>0\) such that \(v=\bar\rho\,\nabla\fe(h_\mu(v))\). By~\eqref{eq:JdeV_nosup}, \(J(v/\bar\rho) = \spr{h_\mu(v)}{v/\bar\rho} - \mu\), and thus
  \(\bar\rho\bigl(\mu+J(v/\bar\rho)\bigr) = \spr{h_\mu(v)}{v} = \icl_\mu(v)\).
  Hence the minimum is achieved at \(\bar\rho\). \(\bar\rho\) can now be identified using~\eqref{eq:condition}: \(\bar\rho = \spr{\partial_\mu h_\mu(v)}{v} = \partial_\mu\icl_\mu(v)\).

  Finally, the minimizer is unique when \(v\neq 0\). Indeed, equality in the Legendre representation of \(J(v/\rho)\) with the dual point \(h_\mu(v)\) requires \(v/\rho=\nabla\fe(h_\mu(v))\). Since \(v\neq 0\), this identity determines \(\rho\) uniquely.
\end{proof}

%%%%%%%%%%%%%%%%%%%%%%%
\section{Two logarithmic estimates}
\label{sec:logAsymptotics}
%%%%%%%%%%%%%%%%%%%%%%%

In this section, we present rather rough logarithmic asymptotics for the fixed-length, fixed-displacement ensemble, both for the standard partition function and for its variant restricted to paths satisfying a suitable confinement condition. These results will be needed in the proof of our main theorems in Section~\ref{sec:ProofMain}.

\subsection{Upper bound on \texorpdfstring{\(\Z{}{n,x}\)}{Znx}}
The first result we need is a uniform upper bound on \(\Z{}{n,x}\), valid for all compatible pairs \((n,x)\), that is, for all \(n\geq 0\) and \(x\in\bbZ^2\) such that \(n-\normI{x}\) is nonnegative and even.
\begin{lemma}\label{lem:logAsymptotics:UpperZ}
For every \(\delta>0\), there exists \(n_\delta\) such that, for all compatible pairs \((n,x)\) with \(n \geq n_\delta\),
\[
  \Z{}{n,x} \leq \exp\{-nJ(x/n)+\delta n\}.
\]
\end{lemma}
\begin{proof}
Let \(B_1\defby\{v\in\bbR^2:\normI v\leq 1\}\). Since \(J(v)=\sup_{h\in\bbR^2}\{\spr{h}{v}-\fe(h)\}\) and \(J\) is finite and continuous on \(B_1\), compactness implies that, for every \(\delta>0\), there exist finitely many vectors \(h^{(1)},\ldots,h^{(m)}\in\bbR^2\) such that
\[
  \forall v\in B_1,\qquad
  J(v) \leq \max_{1\leq i\leq m} \{\spr{h^{(i)}}{v}-\fe(h^{(i)})\} + \frac{\delta}{2} .
\]
For each \(i\), \(\lim_{n\to\infty} \frac1n\log\Z{h^{(i)}}{n} = \fe(h^{(i)})\).
Since there are only finitely many \(i\)'s, there exists \(n_\delta\) such that, for all \(n\geq n_\delta\) and all \(1\leq i\leq m\),
\[
  \Z{h^{(i)}}{n} \leq \exp\{n\fe(h^{(i)})+\tfrac{\delta}{2}n\}.
\]
Let now \((n,x)\) be compatible, \(n\geq n_\delta\), and set
\(v=x/n\). Choose \(i\) such that
\[
  \spr{h^{(i)}}{v}-\fe(h^{(i)})
  \geq J(v)-\frac{\delta}{2}.
\]
Then
\begin{align*}
  \Z{}{n,x}
  &=
  e^{-\spr{h^{(i)}}{x}} \sum_\gamma \IF{\abs{\gamma}=n,\ X(\gamma)=x} e^{\spr{h^{(i)}}{X(\gamma)}}\W(\gamma) \\
  &\leq e^{-\spr{h^{(i)}}{x}}\Z{h^{(i)}}{n} \\
  &\leq \exp\left\{ -n\spr{h^{(i)}}{v} + n\fe(h^{(i)}) + \frac{\delta}{2}n \right\} \\
  &\leq \exp\{-nJ(v)+\delta n\}.
  \qedhere
\end{align*}
\end{proof}

\subsection{Lower bound on a diamond-confined version}

\begin{figure}[t]
  \centering
  \begin{tikzpicture}[scale=0.9,baseline={([yshift=-.5ex]current bounding box.center)}]
      \tikzmath{
          \thet = 35;
          \Rds = 15;
          \lagl = \thet / 2 - 45;
          \uagl = \thet / 2 + 45;
      }
      \useasboundingbox (-.5,-1) rectangle (5,3);
      \clip  (-.5,-1) rectangle (5,3);
      \fill[yellow!30!white] (\lagl:\Rds) -- (0,0) -- (\uagl:\Rds) -- (\Rds,7) -- (\Rds,-3);
      \draw[lightgray] (-1,0) -- (10,0);
      \draw[lightgray] (0,-3) -- (0,7);
      \draw[->, thick] (0,0) -- (\thet:2) node[above] {\(\hat x\)};
      \draw (\lagl:\Rds) -- (0,0) -- (\uagl:\Rds);

      \draw (.7,0) arc[start angle=0, end angle=\thet, x radius=.7, y radius=.7] node[pos=0.5, right] {\(\theta_x\)};
      \draw (4.5,-.75) node {\(\fcone_{\hat x}\)};
    \end{tikzpicture}
    \hspace{1.5cm}
    \begin{tikzpicture}[scale=0.9,baseline={([yshift=-.5ex]current bounding box.center)}]
      \tikzmath{
          \thet = 14;
          \lagl = \thet / 2 - 45;
          \uagl = \thet / 2 + 45;
          \Rds = 15;
      }

      \def\eps{0.1cm}
      \def\mypath{(0,0)}
      \def\parseword#1{%
        \ifx#1\empty
        \else
          \ifx#1U\xdef\mypath{\mypath -- ++(0,\eps)}\fi
          \ifx#1D\xdef\mypath{\mypath -- ++(0,-\eps)}\fi
          \ifx#1R\xdef\mypath{\mypath -- ++(\eps,0)}\fi
          \ifx#1L\xdef\mypath{\mypath -- ++(-\eps,0)}\fi
          \expandafter\parseword
        \fi
      }
      \parseword RRURURDRDRURDDLLDRRRDRRDDRUURRULURRURRRDRRUULURRRRRDLDRRRURRDRRULURRURRUULURURRRUURRURRR\empty

      %       \useasboundingbox (-.5,-1) rectangle (5,3);
      \fill[yellow!30!white] (0,0) -- ++(\lagl:2.54) -- ++(\uagl:3.25) -- ++(180+\lagl:2.54) -- ++(180+\uagl:3.25);
      \draw (0,0) -- ++(\lagl:2.54) -- ++(\uagl:3.25) -- ++(180+\lagl:2.54) -- ++(180+\uagl:3.25);

      \draw[red] \mypath coordinate (x);

      \fill (0,0) circle(.05) node[left] {\(0\)};
      \fill (x) circle(.05) node[right] {\(x\)};
  \end{tikzpicture}
  \caption{Left: The forward cone \(\fcone_{\hat x}\) associated with a direction \(\hat x\) with \(\theta_x\in [0,\pi/2)\). Right: A path contributing to the partition function \(\D{}{n,x}\) (with \(n=88\)).}
  \label{fig:cone}
\end{figure}
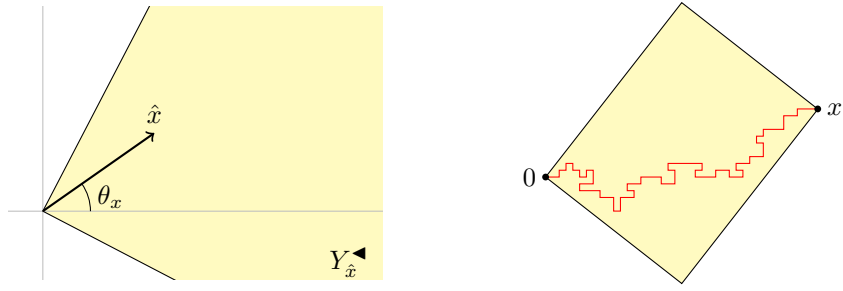

Let \(x\in\bbZ^2\setminus\{0\}\) be such that \(\hat x \defby x/\normII{x} = (\cos\theta_x,\sin\theta_x)\) with \(\theta_x \in [0,\pi/2)\). We define the cones (see Fig.~\ref{fig:cone}, left)
\begin{gather*}
	\fcone_{\hat x} \defby \setof[\big]{(t\cos\theta,t\sin\theta)}{t > 0,\, \tfrac12\theta_x-\tfrac14\pi \leq \theta \leq \tfrac12 \theta_x +\tfrac14\pi}
	\quad\text{and}\quad
	\fbone_{\hat x} \defby -\fcone_{\hat x}.
\end{gather*}
Note that, by construction, both \(x\) and \((1,0)\) belong to the interior of \(\fcone_{\hat x}\) and \(\fcone_{\hat x} \subset \setof{y\in\bbR^2}{y_1>0}\).
We say that a polymer \(\gamma:0 \to x\) is diamond-confined if (see Fig.~\ref{fig:cone}, right)
\[
  \gamma\setminus\{\gamma_0,\gamma_n\} \subset \fcone_{\hat x}\cap(\gamma_n+\fbone_{\hat x}).
\]
In particular, a diamond-confined path cannot revisit either of its endpoints.
The remaining directions \(\hat x\) are treated by the symmetries of the lattice, with the obvious modification of the cones.

The object of interest is the partition function restricted to diamond-confined paths,
\[
  \D{}{n,x} \defby \sum_{\gamma:\,0\to x} \IF{\abs{\gamma}=n} \IF{\gamma \text{ diamond-confined}} \W(\gamma).
\]
The reason for introducing \(\D{}{n,x}\) is that, in the lower bound of Section~\ref{sec:ProofMain}, we will need to concatenate many polymer pieces.
The diamond constraint guarantees that these pieces can be concatenated without creating unwanted intersections, while the next lemma shows that, on the exponential scale relevant here, this restriction has no cost.
\begin{lemma}\label{lem:LogAsymptotics:LowerD}
Let \(K\) be a compact subset of \(\setof{v\in\bbR^2}{0<\normI v<1}\).
Then, for every \(\delta>0\), there exists \(n_\delta\) such that, for all compatible pairs \((n,x)\) with \(n\geq n_\delta\) and \(x/n\in K\),
\[
  \D{}{n,x} \geq \exp\left\{-nJ(x/n)-\delta n\right\}.
\]
\end{lemma}
\begin{proof}
We explain how the estimate follows from the Ornstein--Zernike analysis of~\cite{Ioffe+Velenik-2008}. In fact, the results of that paper imply sharp asymptotics for the diamond-confined partition function; we only extract the logarithmic lower bound needed here.

Given a path \(\gamma=(\gamma_0,\dots,\gamma_\ell)\), we say that \(\gamma_k\) is a cone-point of \(\gamma\) if
\[
  \{\gamma_0,\dots,\gamma_{k-1}\} \subset \gamma_k+\fbone_{\hat x},
  \qquad
  \{\gamma_{k+1},\dots,\gamma_\ell\} \subset \gamma_k+\fcone_{\hat x}.
\]
The path is called backward-irreducible, forward-irreducible or irreducible according to whether the only cone-points are, respectively, the terminal point, the initial point, or the two endpoints. We denote the corresponding classes of paths started at \(0\) by \(\Omega_{\mathrm L}\), \(\Omega_{\mathrm R}\) and \(\Omega\).

For a set \(A\) of paths started at \(0\), write
\[
  \Z{}{n,x}[\gamma\in A] \defby \sum_\gamma \IF{\abs{\gamma}=n,\ X(\gamma)=x,\ \gamma\in A} \W(\gamma).
\]
In particular, \(\D{}{n,x}=\Z{}{n,x}[\gamma\in\calD]\), where \(\calD\) denotes the set of diamond-confined paths.

Let \(v=x/n\in K\), and let \(h=h(v)\) be the unique dual parameter such that \(\nabla\fe(h)=v\). Set \(\lambda=\fe(h)>\lambda_c\). Since \(K\) is compact, the corresponding parameters remain in a compact subset of the supercritical regime.

The cone-renewal decomposition of \cite[Section~3.3.6]{Ioffe+Velenik-2008} gives finite positive boundary measures \(\QL\) and \(\QR\), and a probability measure \(\QD\) on irreducible pieces, all with exponential tails, such that the following representation holds uniformly for \(v=x/n\in K\). For every set \(A\) of paths,
\begin{multline}
  e^{\spr{h}{x}-\lambda n}\Z{}{n,x}[\gamma\in A] \\
  = \sum_{m\geq c_2\normII{x}}
  \QL \times (\QD)^{\times m} \times \QR \Bigl( X(\gamma)=x,\, \abs{\gamma}=n,\, \gamma\in A \Bigr) \\
  + \sfO(e^{-c_1\normII{x}}),
  \label{eq:OZdecompD}
\end{multline}
where \(c_1,c_2>0\) depend only on \(K\), and where \(\gamma = \gamma^{\mathrm L} \concat \gamma^1 \concat \cdots \concat \gamma^m \concat \gamma^{\mathrm R}\) is the concatenation of the sampled pieces.

We apply this representation to diamond-confined paths. It is enough to restrict the right-hand side of~\eqref{eq:OZdecompD} to the event that both boundary pieces are themselves irreducible at both endpoints: \(\gamma^{\mathrm L}\in\Omega\) and \(\gamma^{\mathrm R}\in\Omega\). On this event, the cone structure of the renewal decomposition implies that the full concatenated path is diamond-confined. Hence the contribution of this event gives a lower bound on \(\D{}{n,x}\).

The boundary measures assign strictly positive mass to this event, uniformly for \(v\in K\). This follows from their explicit construction in~\cite{Ioffe+Velenik-2008}: admissible finite irreducible pieces have positive weight, and the set of pieces irreducible at both endpoints is nonempty. Since the corresponding parameters range over a compact subset of the supercritical regime, this positivity is uniform on \(K\).

After this restriction of the boundary pieces, the bulk pieces are still governed by the same i.i.d.\ renewal process with law \(\QD\). The local limit estimate for this renewal process, as in the proof of \cite[(3.36)]{Ioffe+Velenik-2008}, gives, uniformly for \(x/n\in K\),
\[
  e^{\spr{h}{x}-\lambda n}\D{}{n,x} \geq \frac{c_K}{n}
\]
for all sufficiently large \(n\), with \(c_K>0\).

Finally, our choices of \(h\) and \(\lambda\) imply that \(\spr{h}{x}-\lambda n = nJ(x/n)\), by~\eqref{eq:JdeV_nosup}. Therefore
\[
  \D{}{n,x} \geq \frac{c_K}{n}\exp\{-nJ(x/n)\}.
\]
Absorbing the subexponential factor \(c_K/n\) into \(e^{-\delta n}\), for all sufficiently large \(n\), yields the desired estimate.
\end{proof}

%%%%%%%%%%%%%%%%%%%%%%%
\section{Identification of the minimizing curve}
\label{sec:VP}
%%%%%%%%%%%%%%%%%%%%%%%

In this section, we prove the variational statements announced in Section~\ref{sec:Intro:MainResults}.

\medskip
The proof has two complementary parts. The first one is constructive. We begin by deriving, under the temporary assumption that a smooth stationary curve exists, the graph representation and the associated first-order relations. After passing to dual variables and then to monomer time, these relations reduce to a two-dimensional shooting problem. The unique solution of this shooting problem constructs an analytic curve satisfying the endpoint and microscopic-length constraints and hence removes the conditional nature of the preceding derivation. The second part is variational: the convexity of the monomer-time functional shows that this curve is a global minimizer, and the strict convexity of \(J\) identifies the equality case, yielding uniqueness of the minimizing trace among absolutely continuous competitors, up to increasing absolutely continuous reparameterization.

\medskip
The section is organized as follows.
We first show that smooth solutions of the Euler--Lagrange equation may be represented as graphs (Section~\ref{sec:VP:graph}). We then derive the first-order relations satisfied by smooth stationary graphs, expressed in terms of the dual variables associated with the level sets of the free energy (Section~\ref{sec:VP:DualVarFirstOrderRels}). The key step is the introduction of the monomer-time parameterization, in which the dual variable evolves affinely (Section~\ref{sec:VP:MonomerTime}). This reduces the construction of stationary curves to a two-dimensional shooting problem (Section~\ref{sec:VP:shooting}), whose solution is shown to exist and be unique (Section~\ref{sec:VP:Existence+Uniqueness}).
We then prove that the stationary curve is a global minimizer of the variational problem (Section~\ref{sec:VP:Minimizer}), and analyze the equality case to obtain uniqueness of minimizing traces up to increasing absolutely continuous reparameterization (Section~\ref{sec:VP:GeneralUniqueness}).
At this stage, one can show that the effective inverse correlation length parameter along the stationary curve remains bounded away from \(\lambdac\) (Section~\ref{sec:VP:CriticalBarrier}). All these results are then collected to prove Theorem~\ref{thm:SolutionPrimalVP} (Section~\ref{sec:VP:ProofSolutionPrimalVP}).
We then establish equivalence between the primal and dual variational problems (Section~\ref{sec:VP:PrimalDual}), as well as stability of the minimizer (Section~\ref{sec:VP:Stability}).
Finally, we record the explicit solution in the case of the simple random walk (Section~\ref{sec:VP:SRW}) and prove the small-field comparison with the classical catenary (Section~\ref{sec:VP:Smallg}), and study the large-field behavior (Section~\ref{sec:VP:Largeg}).

\subsection{Heuristic motivation for the shooting approach}

Although natural in hindsight, the setup required for the shooting approach used in the first part of this section might not seem obvious at first sight. In this brief subsection, we describe the heuristic argument, based on an analogy with classical mechanics, that led to this formulation.

\medskip
In the case of a particle in a field, Newton's second law yields
\[
\dot{p_1}=0, \qquad \dot{p_2}=g,
\]
and the equations of motion are derived by integrating these identities and expressing the position in terms of the momentum.

The momentum can also be defined through the Lagrangian, by the equation \(p=\nabla_{\dot{q}}L(q,\dot{q})\).
The analog of the momentum in our context is therefore
\[
  p(t) = \nabla_v\icl_{\mu(t)}(\dot\MP(t)) = h_{\mu(t)}(\dot\MP(t)).
\]
We shall show below that under the monomer-time parametrization \(t\mapsto q(t)\) (see~\eqref{eq:Intro:MonomerTime}), the Euler--Lagrange equations recover Newton's second law:
\[
  p'_1(q) = 0, \qquad p'_2(q) = g.
\]
Therefore \(p(q)=(C,s_0+gq)\), for some unknown constants \(C\) and \(s_0\).
Fixing the latter constants determines the entire trajectory \(q\mapsto \eta(q)\).
Since the velocity is given by \(\dot\eta(q) = \nabla \fe(C,s_0+gq)\), one concludes that
\[
  \eta(q) = \int_0^q \nabla \fe(C,s_0+gu)\,\dd u,
\]
with initial conditions \(\eta(0)=(0,0)\) and \(p(0)=(C,s_0)\).

The constants \(C\) and \(s_0\) must of course be chosen in such a way that the trajectory hits the prescribed endpoint at time \(\alpha\): \(\eta(\alpha) = (1,a)\). This means that \(C\) and \(s_0\) must solve the two shooting equations
\[
  \int_0^\alpha \partial_1 \fe(C,s_0+gu)\,\dd u = 1,
  \qquad
  \int_0^\alpha \partial_2 \fe(C,s_0+gu)\,\dd u = a.
\]
We shall see that this system possesses a unique solution.

%%%%
\subsection{Graph representation of stationary curves}
\label{sec:VP:graph}
%%%%

\begin{lemma}\label{lem:graph}
Let \(\lambda\in\bbR\), and let \(\MP:[0,1]\to\bbR^2\) be a nonconstant stationary curve for the Lagrangian
\[
  L(\MP,\dot\MP) = \icl_{\lambda+g\MP_2}(\dot\MP),
\]
with endpoints \(\MP(0)=(0,0)\) and \(\MP(1)=(1,a)\). Assume that \(\lambda+g\MP_2(t)>\lambdac\) on the moving part of the curve and that, after deleting constant pieces, \(\MP\) is smooth and satisfies the Euler--Lagrange equation. Then the trace of \(\MP\) is the graph of a function. More precisely, up to reparameterization, one may write
\[
  \MP(x)=(x,y(x)),\qquad x\in[0,1],
\]
for some function \(y:[0,1]\to\bbR\) with \(y(0)=0\) and \(y(1)=a\).
\end{lemma}

\begin{proof}
After deleting intervals on which \(\dot\MP(t)=0\) a.e.\ and reparameterizing the remaining curve, we may assume that \(\MP\) is parameterized at constant speed, with \(\abs{\dot\MP(t)}=1\) for a.e.\ \(t\). This operation does not change the trace, and it preserves the stationarity condition on the moving part of the curve.

On this representative, the Lagrangian is smooth and the classical Euler--Lagrange equation applies. Let
\[
  h(t) \defby \nabla_v\icl_{\lambda+g\MP_2(t)}(\dot\MP(t)).
\]
Since the Lagrangian does not depend explicitly on the first coordinate, the first Euler--Lagrange equation gives
\[
  \frac{\dd}{\dd t}h_1(t)=0,
\]
so that \(h_1(t)\equiv C\) for some constant \(C\).

On the set where \(\dot\MP(t)\ne0\),
\[
  \operatorname{sign} \bigl(\partial_{v_1}\icl_\mu(v)\bigr)
  =
  \operatorname{sign}(v_1).
\]
Indeed, let \(\tilde h=\nabla_v\icl_\mu(v)\). By convex duality, \(v\) is
a positive multiple of \(\nabla\fe(\tilde h)\). Since \(\fe\) is even in
its first coordinate and \(\partial_{11}\fe>0\) away from the origin,
\[
  \operatorname{sign}(\partial_{v_1}\icl_\mu(v))
  =
  \operatorname{sign}(\tilde h_1)
  =
  \operatorname{sign}(\partial_1\fe(\tilde h))
  =
  \operatorname{sign}(v_1).
\]

Since \(h_1(t)\equiv C\), it follows that \(\dot\MP_1(t)\) has a fixed sign on the moving part. The total horizontal displacement is equal to \(1\), hence \(C>0\) and \(\dot\MP_1(t)>0\) a.e.\ on the moving part. Therefore \(\MP_1\) is strictly increasing there. In particular, it cannot take the same value at two distinct times unless \(\dot\MP=0\) a.e.\ on the corresponding interval, which has been removed. Hence the trace of \(\MP\) is the graph of a function.
\end{proof}
\begin{remark}
The removal of constant pieces allows one to avoid the lack of differentiability of the Lagrangian at \(v=0\). An alternative approach would consist in working directly with the corresponding subgradient formulation of the Euler--Lagrange equation.
\end{remark}

%%%%
\subsection{Dual variables and first-order relations}
\label{sec:VP:DualVarFirstOrderRels}
%%%%

We now write the Euler--Lagrange equation in a form adapted to the geometry of the level sets of \(\fe\). Let \(\MP(x)=(x,y(x))\) be a smooth graph satisfying the Euler--Lagrange equation, and set \(\mu(x) = \lambda+gy(x)\) and \(v(x)=(1,y'(x))\).
Introduce the dual variable \(h(x)\defby \nabla_v\icl_{\mu(x)}(v(x))=(C,s(x))\).

We shall repeatedly use the following elementary consequence of convex duality, which relates the slope \(v_2\) to the dual variables.
\begin{lemma}\label{lem:duality}
  Let \(\mu>\lambdac\) and \(v=(1,v_2)\). Let \(h_\mu(v)=(C,s)\in\partial\K_\mu\) be the dual point associated with
the direction \(v\). Then
  \[
    v_2=\frac{\partial_2 \fe(C,s)}{\partial_1 \fe(C,s)},
    \qquad\text{and}\qquad
    \partial_\mu\icl_\mu(1,v_2)=\frac{1}{\partial_1 \fe(C,s)}.
  \]
\end{lemma}

\begin{proof}
  By the definition of \(h_\mu(v)\), we have \(\icl_\mu(v)=\spr{h_\mu(v)}{v}\) and \(\fe(h_\mu(v)) = \mu\).
  Let \(b > 0\) be such that \(v = b\nabla \fe(h)\). Since the first component of \(v\) is equal to \(1\), we deduce that
  \[
    b = \frac{1}{\partial_1 \fe(C,s)},
  \]
  and the first claim follows. Differentiating \(\icl_\mu(v)=\spr{h_\mu(v)}{v}\) with respect to \(\mu\) gives
  \[
    \partial_\mu\icl_\mu(v) = \spr{\partial_\mu h_\mu(v)}{v} = b\, \spr{\partial_\mu h_\mu(v)}{\nabla \fe(h_\mu(v))} = b,
  \]
  since differentiating \(\fe(h_\mu(v)) = \mu\) yields \(\spr{\nabla \fe(h_\mu(v))}{\partial_\mu h_\mu(v)} = 1\). This proves the second claim.
\end{proof}

Since \(h(x)\in\partial \K_{\mu(x)}\), we have
\begin{equation}\label{eq:level-set-constraint}
  \fe(C,s(x)) = \lambda+gy(x).
\end{equation}
Moreover, Lemma~\ref{lem:duality} gives
\begin{equation}\label{eq:yprime}
  y'(x)=\frac{\partial_2\fe(C,s(x))}{\partial_1\fe(C,s(x))}
\end{equation}
and
\begin{equation}\label{eq:dmunumu}
  \partial_\mu\icl_\mu(1,y'(x))
  \vert_{\mu=\lambda+gy(x)}
  =
  \frac1{\partial_1\fe(C,s(x))}.
\end{equation}
The second Euler--Lagrange equation reads
\[
  \frac{\dd}{\dd x}s(x)
  =
  g\,
  \partial_\mu\icl_\mu(1,y'(x))
  \vert_{\mu=\lambda+gy(x)} .
\]
Using~\eqref{eq:dmunumu}, we obtain
\begin{equation}\label{eq:sprime}
  s'(x)=\frac{g}{\partial_1\fe(C,s(x))}.
\end{equation}

These identities reduce the construction of a stationary graph satisfying the endpoint and microscopic-length constraints to the determination of the scalar dual variable \(s(x)\). In the next subsection, we introduce a convenient parameterization in which the evolution equation~\eqref{eq:sprime} becomes affine.

%%%%
\subsection{Monomer-time parameterization}
\label{sec:VP:MonomerTime}
%%%%

Recall that, in graph parameterization, the microscopic-length constraint reads
\[
  \int_0^1 \partial_\mu\icl_\mu(1,y'(x)) \vert_{\mu=\lambda+gy(x)} \,\dd x = \alpha .
\]
By~\eqref{eq:dmunumu},
\[
  \partial_\mu\icl_\mu(1,y'(x)) \vert_{\mu=\lambda+gy(x)} = \frac1{\partial_1\fe(C,s(x))}.
\]
This suggests introducing the monomer-time parameter \(q\) via
\[
  \frac{\dd q}{\dd x} = \frac1{\partial_1\fe(C,s(x))}, \qquad q(0)=0,
\]
Since \(C>0\), this is an increasing change of variables. Of course,
\[
  \frac{\dd x}{\dd q}
  =
  \partial_1\fe(C,s(q)).
\]
The microscopic-length constraint gives \(q(1)=\alpha\).
Using~\eqref{eq:yprime}, we obtain
\[
  \frac{\dd y}{\dd q} = \partial_2\fe(C,s(q)).
\]
We now regard \(s\) as a function of \(q\), writing \(s=s(q)\), and compute the evolution of \(s\) in this parameterization. From~\eqref{eq:sprime} and the chain rule, we obtain
\[
  \frac{ds}{dq} = g.
\]
Therefore, \(s(q)=s_* + gq\) for some constant \(s_*\). Thus, in monomer time, the dual variable evolves affinely.

The physical curve is then recovered by integrating \(\nabla\fe\) along this affine path in dual space:
\begin{equation}\label{eq:parameteric-q}
  x(q) = \int_0^q \partial_1 \fe(C,s_*+g u)\,\dd u,
  \qquad
  y(q) = \int_0^q \partial_2 \fe(C,s_*+g u)\,\dd u.
\end{equation}
Equivalently, using \(s=s_*+gq\), this representation can be written as
\begin{equation}\label{eq:parameteric-solution}
  x(s)
  =
  \frac1g\int_{s_*}^s \partial_1\fe(C,u)\,\dd u,
  \qquad
  y(s)
  =
  \frac{\fe(C,s)-\fe(C,s_*)}{g},
\end{equation}
with \(s\in[s_*,s_*+g\alpha]\).

%%%%
\subsection{The shooting problem}
\label{sec:VP:shooting}
%%%%

The representation~\eqref{eq:parameteric-q} reduces the problem to determining the constants \((C,s_*)\) so that the curve connects \((0,0)\)
to \((1,a)\).
Since the curve is parameterized by \(q\in[0,\alpha]\), the boundary conditions become \(x(\alpha) = 1\) and \(y(\alpha) = a\).
Using~\eqref{eq:parameteric-q}, this yields the system
\begin{equation}\label{eq:shooting}
  \int_0^\alpha \partial_1 \fe(C,s_*+g u)\,\dd u = 1,
  \qquad
  \int_0^\alpha \partial_2 \fe(C,s_*+g u)\,\dd u = a.
\end{equation}
We now analyze this system.

\subsubsection*{Uniqueness}
We first prove uniqueness by reducing the system to a single equation in \(C\), using the monotonicity of the corresponding one-dimensional maps.
Fix \(C>0\). We first consider the second equation in~\eqref{eq:shooting}. Define
\[
  \Phi(C,s) \defby \int_0^\alpha \partial_2 \fe(C,s+gu)\,\dd u .
\]
Since \(\Hess\fe\) is positive definite on \(\mathbb R^2\setminus\{0\}\) and \(C>0\), we have
\[
  \partial_s \Phi(C,s) = \int_0^\alpha \partial_{22} \fe(C,s+gu)\,\dd u > 0.
\]
Thus, for each fixed \(C>0\), the equation \(\Phi(C,s)=a\) has at most one solution, which we denote by \(s=s(C)\).
By the implicit function theorem, \(s(C)\) is \(C^1\) on its domain of definition. Differentiating the identity \(\Phi(C,s(C)) = a\) with respect to \(C\) leads to
\begin{equation}\label{eq:sprimeC}
  s'(C)
  =
  - \frac{
  \displaystyle\int_0^\alpha \partial_{12} \fe(C,s(C)+gu)\,\dd u
  }{
  \displaystyle\int_0^\alpha \partial_{22} \fe(C,s(C)+gu)\,\dd u
  }.
\end{equation}
We now substitute \(s=s(C)\) into the first equation and define
\[
  \Psi(C) \defby \int_0^\alpha \partial_1 \fe(C,s(C)+gu)\,\dd u .
\]
We claim that \(\Psi\) is strictly increasing. Differentiating gives
\[
  \Psi'(C)
  =
  \int_0^\alpha \partial_{11} \fe(C,s(C)+gu)\,\dd u
  +
  s'(C) \int_0^\alpha \partial_{12} \fe(C,s(C)+gu)\,\dd u .
\]
Using \eqref{eq:sprimeC}, this becomes
\[
  \Psi'(C) = A-\frac{B^2}{D} = \frac{AD-B^2}{D},
\]
where
\[
  A \defby \int_0^\alpha \partial_{11} \fe(C,s(C)+gu)\,\dd u,\qquad
  B \defby \int_0^\alpha \partial_{12} \fe(C,s(C)+gu)\,\dd u,
\]
and
\[
  D \defby \int_0^\alpha \partial_{22} \fe(C,s(C)+gu)\,\dd u .
\]
Since \(D>0\), it remains to check that \(AD-B^2>0\). But
\[
  \begin{pmatrix}
    A & B\\
    B & D
  \end{pmatrix}
  =
  \int_0^\alpha
  \Hess\fe(C,s(C)+gu)\, \dd u .
\]
Since \(\Hess\fe\) is positive definite along the whole segment \(\setof{(C,s(C)+gu)}{{0\leq u\leq\alpha}}\), its integral is also positive definite. Therefore \(AD-B^2>0\), and \(\Psi'(C) > 0\).

We conclude that the equation \(\Psi(C)=1\) can have at most one solution. Since, for each such \(C\), the second equation in~\eqref{eq:shooting} determines \(s(C)\) uniquely, the shooting system has at most one solution.

\subsubsection*{Existence}
We now turn to the proof of existence. The argument relies on the following lemma.
\begin{lemma}\label{lem:sC-Psi}
For every \(C>0\), there exists a unique \(s(C)\in\bbR\) such that
\[
  \int_0^\alpha \partial_2 \fe(C,s(C)+gu) \, \dd u = a .
\]
Moreover, \(C\mapsto s(C)\) is analytic. If
\[
  \Psi(C) \defby \int_0^\alpha \partial_1 \fe(C,s(C)+gu) \, \dd u,
\]
then
\[
  \lim_{C\downarrow 0} \Psi(C) = 0
  \quad\text{and}\quad
  \lim_{C\to\infty} \Psi(C) = \alpha - \abs{a}.
\]
\end{lemma}
\begin{proof}
Fix \(C>0\) and define
\[
  \Phi_C(s) \defby \int_0^\alpha \partial_2 \fe(C,s+gu) \, \dd u .
\]
We first identify the range of the function \(s\mapsto\Phi_C(s)\). We claim that, for every fixed \(C>0\),
\[
  \lim_{s\to+\infty} \partial_2 \fe(C,s) = 1
  \quad\text{and}\quad
  \lim_{s\to-\infty} \partial_2 \fe(C,s) = -1.
\]
Indeed, by convexity of the function \(s\mapsto \fe(C,s)\),
\[
  \forall s>0,\qquad
  \partial_2 \fe(C,s) \geq \frac{\fe(C,s)-\fe(C,0)}{s}.
\]
By the lower bound~\eqref{eq:Boundsf}, \(\fe(C,s)\geq s\), so that the right-hand side tends to \(1\) as
\(s\to+\infty\). From this and the elementary bound \(\partial_2 \fe(C,s)\leq 1\), we obtain \(\lim_{s\to+\infty} \partial_2 \fe(C,s) = 1\).
The second limit follows from symmetry in the second coordinate.

We infer, by dominated convergence, that
\[
  \lim_{s\to+\infty}\Phi_C(s)=\alpha,
  \quad\text{and}\quad
  \lim_{s\to-\infty}\Phi_C(s)=-\alpha.
\]
Since \(\abs{a}<\alpha\), the intermediate value theorem gives existence of a (necessarily unique) solution to the equation \(\Phi_C(s)=a\), which we denote by \(s(C)\). Since
\[
  \partial_s \Phi_C(s)
  =
  \int_0^\alpha \partial_{22} \fe(C,s+gu) \, \dd u>0,
\]
and since \(\fe\) is analytic on \(\bbR^2\setminus\{0\}\), the analytic implicit function theorem shows that \(C\mapsto s(C)\) is analytic.

We now establish the limiting behavior of \(\Psi\). First, as \(C\downarrow0\), the symmetry of \(\fe\) in the first coordinate gives \(\partial_1 \fe(0,s) = 0\) for all \(s\in\bbR\). Together with \(\abs{\partial_1 \fe}\leq 1\), dominated convergence yields
\[
  \lim_{C\downarrow 0} \Psi(C) = 0.
\]

We now consider the limit \(C\to\infty\). Observe that \(\normII{(C,s(C)+gu)} \geq C\), and thus \(\normII{(C,s(C)+gu)}\xrightarrow{C\to\infty}\infty\) uniformly in \(u\in[0,\alpha]\).
Hence, by Point~\ref{lem:FreeEnergy:Largeh} of Lemma~\ref{lem:FreeEnergy},
\[
  \normI{\nabla \fe(C,s(C)+gu)} \xlongrightarrow{C\to\infty} 1
\]
uniformly in \(u\in [0,\alpha]\). Since \(C>0\), it follows from symmetry and convexity that \(\partial_1 \fe(C,s)\geq 0\), and therefore \(\partial_1 \fe(C,s(C)+gu) + \abs{\partial_2 \fe(C,s(C)+gu)} = \normI{\nabla \fe(C,s(C)+gu)} \xrightarrow{C\to\infty} 1\) uniformly in \(u\in[0,\alpha]\). Hence
\[
  \Psi(C) + \int_0^\alpha \abs{\partial_2 \fe(C,s(C)+gu)} \, \dd u \xlongrightarrow{C\to\infty} \alpha .
\]
The desired conclusion will thus follow once we show that
\[
  \int_0^\alpha \abs[\big]{\partial_2 \fe(C,s(C)+gu)} \, \dd u \xlongrightarrow{C\to\infty} \abs{a}.
\]
Set \(w_C(u) \defby \partial_2 \fe(C,s(C)+gu)\).
Since \(s\mapsto \partial_2 \fe(C,s)\) is strictly increasing and \(\partial_2 \fe(C,0)=0\), the sign of \(w_C(u)\) is the sign of \(s(C)+gu\).

Assume first that \(a>0\). We claim that, for all sufficiently large \(C\), \(s(C)\geq 0\).
Indeed, if this were not the case, we could find a sequence \(C_n\xrightarrow{n\to\infty}\infty\) with \(s(C_n)<0\) for all \(n\). Since
\begin{equation}\label{eq:intweqa}
  \forall n,\qquad
  \int_0^\alpha w_{C_n}(u) \, \dd u = a > 0,
\end{equation}
the interval \([s(C_n),s(C_n)+g\alpha]\) cannot be contained in \((-\infty,0]\). Hence it must contain \(0\). Therefore
\(\abs{s(C_n)+gu} \leq g\alpha\) for all \(u\in [0,\alpha]\).
Along this sequence, the second coordinate stays bounded while the first coordinate tends to \(+\infty\). Hence
\(\partial_2 \fe(C_n,s(C_n)+gu)\xlongrightarrow{n\to\infty} 0\) uniformly in \(u\in[0,\alpha]\). Consequently,
\[
  \int_0^\alpha w_{C_n}(u) \, \dd u\xlongrightarrow{n\to\infty} 0,
\]
contradicting the identity~\eqref{eq:intweqa}.
Thus \(s(C)\geq 0\) for all large \(C\), and therefore \(w_C(u)\geq 0\) for all \(u\in [0,\alpha]\). It follows that, for all large \(C\),
\[
  \int_0^\alpha \abs{w_C(u)} \, \dd u = \int_0^\alpha w_C(u) \, \dd u = a.
\]

The case \(a<0\) is identical: one shows that, for all large \(C\), \(s(C)+g\alpha\leq 0\), so that \(w_C(u)\leq 0\) on \([0,\alpha]\), and hence
\[
  \int_0^\alpha \abs{w_C(u)} \, \dd u = -\int_0^\alpha w_C(u) \, \dd u = \abs{a}.
\]

Finally, if \(a=0\), uniqueness and symmetry give \(s(C) = -g\alpha/2\).
Thus \(s(C)+gu\) remains bounded uniformly in \(u\), and the same large-\(C\) argument gives \(\sup_{u\in [0,\alpha]} \abs[\big]{\partial_2 \fe(C,s(C)+gu)} \xrightarrow{C\to\infty} 0\).
Therefore
\[
  \int_0^\alpha \abs{w_C(u)} \, \dd u \xrightarrow{C\to\infty} 0 = \abs{a}.
\]
This proves the desired convergence in all cases.
\end{proof}

Since \(\alpha>1+\abs{a}\), it follows from Lemma~\ref{lem:sC-Psi} that
\[
  \lim_{C\downarrow0}\Psi(C)=0<1<\alpha-\abs{a} = \lim_{C\to\infty}\Psi(C).
\]
By continuity, there exists \(C>0\) such that \(\Psi(C)=1\).
Together with the corresponding \(s(C)\), this gives a solution of the shooting system.

%%%%
\subsection{Existence and uniqueness of the stationary solution}
\label{sec:VP:Existence+Uniqueness}
%%%%

Although global uniqueness of minimizers will ultimately follow from strict convexity, the uniqueness of the shooting solution gives a direct and useful identification of the analytic stationary curve.

\begin{proposition}\label{prop:shooting-solution}
There exists a unique pair
\((C_*,s_*)\in(0,\infty)\times\bbR\) solving the shooting system
\[
  \int_0^\alpha \partial_1 \fe(C_*,s_*+gu) \, \dd u = 1,
  \qquad
  \int_0^\alpha \partial_2 \fe(C_*,s_*+gu) \, \dd u = a.
\]
The curve defined by
\[
  x(q) = \int_0^q \partial_1 \fe(C_*,s_*+gu) \, \dd u,
  \qquad
  y(q) = \int_0^q \partial_2 \fe(C_*,s_*+gu) \, \dd u,
\]
with \(q\in [0,\alpha]\), connects \((0,0)\) to \((1,a)\). Moreover,
\[
  \forall q\in [0,\alpha],\qquad
  \partial_1\fe(C_*,s_*+gq) > 0,
\]
so \(q\mapsto x(q)\) is strictly increasing and the curve is the graph of an analytic function \(y_*:[0,1]\to\bbR\).
Finally, if \(\lambda_* \defby \fe(C_*,s_*)\), then the curve \(\MP_*(x)=(x,y_*(x))\) satisfies the Euler--Lagrange equation for the original Lagrangian and the primal microscopic-length constraint. It is the unique smooth graph solution of the Euler--Lagrange equation satisfying the endpoint and mi\-cro\-scop\-ic-length constraints.
\end{proposition}
\begin{proof}
  Existence and uniqueness of \((C_*,s_*)\) follow from the preceding shooting argument. The endpoint conditions follow immediately from the shooting system: \(x(\alpha) = 1\) and \(y(\alpha) = a\).

  Since \(x(q)\) is strictly increasing, the curve can be written as a graph \(y_*(x)\).
  The free energy \(\fe\) is analytic on \(\bbR^2\setminus\{0\}\). Since \(\setof{(C_*,s_*+gq)}{q\in[0,\alpha]}\) does not meet the origin, both \(x(q)\) and \(y(q)\) are analytic functions of \(q\). The analytic inverse function theorem therefore implies that \(y_*\) is analytic as a function of \(x\).

  Finally,
  \begin{equation}\label{eq:stationary-identities}
    y_*'(x(q))
    =
    \frac{\partial_2 \fe(C_*,s_*+gq)}
        {\partial_1 \fe(C_*,s_*+gq)}
    \qquad\text{and}\qquad
    \lambda_*+g y_*(x(q))
    =
    \fe(C_*,s_*+gq).
  \end{equation}
  The first identity follows directly from the definitions of \(x(q)\) and \(y(q)\). For the second one, both sides have derivative \(g\,\partial_2\fe(C_*,s_*+gq)\) with respect to \(q\), and they agree at \(q=0\).
  These identities show in particular that the constructed graph solves the original Euler--Lagrange equation. Indeed, in the parameterization \(\MP_*(x)=(x,y_*(x))\), one has
  \[
    \nabla_v\icl_{\lambda_*+gy_*(x)}(1,y_*'(x))
    =
    (C_*,s_*+gq(x)),
  \]
  where \(q=q(x)\) is the inverse of \(x(q)\). The first component is
  constant. Moreover,
  \[
    \frac{\dd}{\dd x}(s_*+gq(x))
    =
    \frac{g}{\partial_1\fe(C_*,s_*+gq(x))}
    =
    g\,\partial_\mu\icl_\mu(1,y_*'(x))
    \vert_{\mu=\lambda_*+gy_*(x)},
  \]
  where the last equality follows from Lemma~\ref{lem:duality}. These are precisely the two Euler--Lagrange equations for the original Lagrangian. Thus the shooting construction removes the conditional assumption made in the derivation of the first-order relations.

  It remains to check the primal microscopic-length constraint. Since
  \[
    \partial_\mu\icl_\mu(1,y_*'(x(q))) \vert_{\mu=\lambda_*+gy_*(x(q))}
    =
    \frac{1}{\partial_1 \fe(C_*,s_*+gq)},
  \]
  we obtain, using \(\dd x = \partial_1 \fe(C_*,s_*+gq)\,\dd q\),
  \[
    \int_0^1 \partial_\mu\icl_\mu(1,y_*'(x)) \vert_{\mu=\lambda_*+gy_*(x)} \,\dd x
    =
    \int_0^\alpha \dd q
    =
    \alpha.
  \]
  Thus the primal microscopic-length constraint is satisfied.

  Conversely, any smooth graph solution of the Euler--Lagrange equation satisfying the endpoint and microscopic-length constraints gives, by the first-order relations derived above, a pair of constants \((C,s_0)\) solving the shooting system. By uniqueness of the solution to the shooting system, we must have \((C,s_0)=(C_*,s_*)\), and the corresponding graph coincides with \(y_*\).
\end{proof}

%%%%
\subsection{The stationary solution is a global minimizer}
\label{sec:VP:Minimizer}
%%%%

Proposition~\ref{prop:shooting-solution} identifies the unique smooth stationary candidate. We now show that the stationary solution constructed above is in fact a global minimizer. It will be convenient to work with the monomer-time functional
\[
  \calB(\eta) \defby \int_0^\alpha \bigl[ J(\eta'(q)) + g\eta_2(q) \bigr]\, \dd q,
\]
defined on absolutely continuous curves
\[
  \eta:[0,\alpha]\to\mathbb R^2,
  \qquad
  \eta(0)=(0,0),\quad \eta(\alpha)=(1,a).
\]

Let \((C_*,s_*)\) be the unique solution of the shooting system and set
\[
  h_*(q) \defby (C_*,s_*+gq),
  \qquad
  \eta_*(q) = \int_0^q \nabla \fe(C_*,s_*+gu) \, \dd u.
\]
By construction, \(\eta_*(0)=(0,0)\) and \(\eta_*(\alpha)=(1,a)\).
We first prove that \(\eta_*\) minimizes \(\mathcal B\). Since \(J\) is the Legendre transform of \(\fe\), and since \(\eta_*'(q)=\nabla \fe(h_*(q))\), we have \(h_*(q)=\nabla J(\eta_*'(q))\).
Therefore, by convexity of \(J\), for every admissible curve \(\eta\),
\[
  J(\eta'(q)) \geq J(\eta_*'(q)) + \spr[\big]{h_*(q)}{\eta'(q)-\eta_*'(q)}
\]
for a.e.\ \(q\). Integrating and adding the gravitational term gives
\begin{align*}
  \calB(\eta)-\calB(\eta_*)
  &\geq
  \int_0^\alpha \spr[\big]{h_*(q)}{\eta'(q)-\eta_*'(q)} \, \dd q
  + g\int_0^\alpha \bigl(\eta_2(q)-\eta_{*,2}(q)\bigr) \, \dd q .
\end{align*}
Integration by parts yields
\begin{multline*}
  \int_0^\alpha \spr[\big]{h_*(q)}{\eta'(q)-\eta_*'(q)} \, \dd q
  =
  \Bigl[ \spr[\big]{h_*(q)}{\eta(q)-\eta_*(q)} \Bigr]_0^\alpha \\
  - \int_0^\alpha \spr[\big]{h_*'(q)}{\eta(q)-\eta_*(q)} \, \dd q .
\end{multline*}
The boundary term vanishes because \(\eta\) and \(\eta_*\) have the same endpoints, while
\[
  \bigl\langle h_*'(q),\eta(q)-\eta_*(q)\bigr\rangle
  =
  g\bigl(\eta_2(q)-\eta_{*,2}(q)\bigr),
\]
since \(h_*'(q)=(0,g)\). We conclude that \(\calB(\eta)-\calB(\eta_*) \geq 0\), and thus that \(\eta_*\) is a global minimizer of \(\calB\).

\medskip
Let us now relate this to the primal variational problem. Let \((\lambda,\MP)\) be any primal-admissible pair, and define \(\rho(t) \defby \partial_\mu \icl_\mu(\dot\MP(t)) \vert_{\mu=\lambda+g\MP_2(t)}\).
By the primal constraint,
\[
  \int_0^1 \rho(t) \, \dd t = \alpha.
\]
Define the monomer-time variable
\[
  q(t) \defby \int_0^t \rho(u) \,\dd u,
\]
and let \(\eta\) be the associated parameterization, so that \(\eta(q(t))=\MP(t)\).
As this step will be repeated several times in the sequel, we will refer to this construction as the monomer-time representation associated with \((\lambda,\MP)\). Then,
\begin{align*}
  \mathcal A(\MP,\lambda)
  &= \int_0^1 \icl_{\lambda+g\MP_2(t)}(\dot\MP(t))\, \dd t - \lambda\alpha \\
  &= \int_0^1 \bigl[ \lambda+g\MP_2(t) + J(\dot\MP(t)/\rho(t)) \bigr] \rho(t)\, \dd t - \lambda\alpha \\
  &= \int_0^1 \bigl[ g\eta_2(q(t)) + J(\dot\eta(q(t))) \bigr] \rho(t)\, \dd t \\
  &= \int_0^\alpha \bigl[ g\eta_2(q) + J(\dot\eta(q)) \bigr] \dd q
  = \mathcal B(\eta),
\end{align*}
where the second line follows from the perspective identity of Lemma~\ref{lem:perspective}.
Since \(\eta_*\) minimizes \(\mathcal B\), it follows that
\(\calA(\MP,\lambda) = \calB(\eta) \geq \calB(\eta_*)\).

It remains to identify \(\calB(\eta_*)\) with the value of the primal functional at the stationary solution.
Let \(\lambda_* \defby \fe(C_*,s_*)\). Since
\[
  \eta_{*,2}(q) = \int_0^q \partial_2 \fe(C_*,s_*+gu) \, \dd u,
\]
we have \(\lambda_*+g\eta_{*,2}(q) = \fe(C_*,s_*+gq) = \fe(h_*(q))\).
Together with \(\eta_*'(q)=\nabla \fe(h_*(q))\), this shows that equality holds in the perspective identity with \(\rho=1\), \(\mu=\lambda_* + g\eta_{*,2}(q)\), and \(v=\eta_*'(q)\).
Therefore \(\calA(\MP_*,\lambda_*) = \calB(\eta_*)\).
Combining the preceding inequalities gives, for every primal-admissible pair \((\lambda,\MP)\),
\[
  \calA(\MP,\lambda) \geq \calA(\MP_*,\lambda_*).
\]
Thus the stationary solution is a global minimizer of the primal variational problem.

%%%%
\subsection{Uniqueness of minimizers up to reparameterization}
\label{sec:VP:GeneralUniqueness}
%%%%

We now analyze the equality case in the previous argument. This yields uniqueness of the minimizing trace among all absolutely continuous competitors.

\begin{proposition}\label{prop:AC-uniqueness}
Let \((\lambda,\MP)\) be a primal minimizer. Then \(\lambda=\lambda_*\), and the trace of \(\MP\) coincides with the trace of the curve \(\eta_*\) constructed above. More precisely, after removing constant pieces, there exists a nondecreasing absolutely continuous map
\(\theta:[0,1]\to[0,\alpha]\), satisfying \(\theta(0)=0\) and \(\theta(1)=\alpha\), such that \(\MP(t) = \eta_*(\theta(t))\).

Conversely, every such reparameterization, together with \(\lambda=\lambda_*\), is a primal minimizer.
\end{proposition}

\begin{proof}
Let \((\lambda,\MP)\) be a primal minimizer, and let \(\eta\) be the monomer-time representation associated with \((\lambda,\MP)\), after
removing constant pieces if necessary. By the perspective identity, \(\calA(\MP,\lambda)=\calB(\eta)\).
Since \((\lambda,\MP)\) is minimizing and since the previous subsection proved that \(\inf\calA=\calB(\eta_*)\), we get \(\calB(\eta)=\calB(\eta_*)\).
Thus \(\eta\) is also a minimizer of \(\calB\).

We now use the equality case in the convexity argument proving minimality of \(\eta_*\). Recall that \(h_*(q)=(C_*,s_*+gq)\) and \(\eta_*'(q)=\nabla \fe(h_*(q))\), and hence \(h_*(q)=\nabla J(\eta_*'(q))\). For every admissible \(\eta\), convexity of \(J\) gives
\[
  J(\eta'(q))
  \geq
  J(\eta_*'(q))
  +
  \spr[\big]{h_*(q)}{\eta'(q)-\eta_*'(q)}
\]
for a.e.\ \(q\). In the proof of minimality, after integration by parts, this inequality yielded \(\calB(\eta)-\calB(\eta_*) \geq 0\).
Since equality holds, equality must hold in the above convexity inequality for a.e.\ \(q\). By strict convexity of \(J\) on the interior of its effective domain, this implies \(\eta'(q) = \eta_*'(q)\) for a.e.\ \(q\in [0,\alpha]\).
Since \(\eta(0)=\eta_*(0)\), we conclude that
\[
  \forall q\in[0,\alpha],\qquad
  \eta(q)=\eta_*(q).
\]
Therefore the trace of \(\MP\) is the trace of \(\eta_*\). Equivalently, \(\MP(t)=\eta_*(q(t))\), where \(q\) is the monomer-time change of variables associated with \((\lambda,\MP)\). This proves uniqueness of the minimizing trace.

It remains to identify the value of \(\lambda\). On the moving part of the curve, equality in the perspective identity implies that the local parameter satisfies \(\lambda + g\eta_{*,2}(q) = \fe(h_*(q))\).
But for the constructed minimizer we have \(\fe(h_*(q)) = \lambda_* + g\eta_{*,2}(q)\).
Hence \(\lambda = \lambda_*\).

Conversely, let \(\theta:[0,1]\to[0,\alpha]\) be nondecreasing and absolutely continuous, with \(\theta(0)=0\) and \(\theta(1)=\alpha\), and set \(\MP(t) \defby \eta_*(\theta(t))\).
Since the primal functional and the primal constraint are homogeneous of degree one in the velocity, this reparameterization does not change the value of the action. More explicitly, along \(\eta_*\) the natural
monomer-time density is equal to \(1\), and by homogeneity, along \(\MP\) it is equal to \(\theta'(t)\). Hence
\[
  \int_0^1 \partial_\mu\icl_\mu(\dot\MP(t)) \vert_{\mu=\lambda_*+g\MP_2(t)} \, \dd t
  =
  \int_0^1 \theta'(t) \, \dd t
  =
  \alpha.
\]
Similarly, \(\calA(\MP,\lambda_*) = \calA(\MP_*,\lambda_*)\).
Thus every such reparameterization is a primal minimizer.
\end{proof}

%%%%
\subsection{Critical barrier}
\label{sec:VP:CriticalBarrier}
%%%%

The next result shows that, along the stationary solution, the local parameter \(\lambda + g\eta_{*,2}(q)\) remains bounded away from the critical value \(\lambdac\).
\begin{lemma}\label{lem:uniform-gap}
There exists \(\delta_*>0\) such that
\[
  \forall q\in [0,\alpha],\qquad
  \lambda_*+g\eta_{*,2}(q) \geq \lambdac+\delta_*.
\]
Equivalently, in graph parameterization,
\[
  \forall x\in [0,1],\qquad
  \lambda_* + g y_*(x) \geq \lambdac+\delta_*.
\]
\end{lemma}

\begin{proof}
By~\eqref{eq:stationary-identities},
\[
  \lambda_*+g\eta_{*,2}(q)
  =
  \fe(C_*,s_*+gq).
\]
Since \(C_*>0\), the set \(\Gamma_* \defby \setof{(C_*,s_*+gq)}{q\in[0,\alpha]}\) is compact and does not contain the origin.
By Lemma~\ref{lem:LevelSets}, \(\K_{\lambdac} = \{0\}\). Hence \(\fe(h)>\lambdac\) for all \(h\in\Gamma_*\).
By continuity of \(\fe\) and compactness of \(\Gamma_*\), the minimum
\[
  \delta_* \defby \min_{q\in[0,\alpha]} \bigl\{\fe(C_*,s_*+gq)-\lambdac\bigr\}
\]
is strictly positive. Using the identity above, this gives
\[
  \forall q\in [0,\alpha],\qquad
  \lambda_*+g\eta_{*,2}(q) \geq \lambdac+\delta_*.
\]

Finally, since the graph parameterization is obtained from the same trace, \(y_*(x)=\eta_{*,2}(q(x))\) for the inverse \(q=q(x)\) of \(x(q)=\eta_{*,1}(q)\). The same lower bound therefore holds for \(\lambda_*+g y_*(x)\) on \([0,1]\).
\end{proof}

%%%%
\subsection{Gathering the pieces: proof of Theorem~\ref{thm:SolutionPrimalVP}}
\label{sec:VP:ProofSolutionPrimalVP}
%%%%

\begin{proof}[Proof of Theorem~\ref{thm:SolutionPrimalVP}]
Proposition~\ref{prop:shooting-solution} gives the existence and uniqueness of the pair \((C_*,s_*)\) solving the shooting system. It also constructs the corresponding curve
\[
  \eta_*(q)
  =
  \int_0^q \nabla\fe(C_*,s_*+gu)\,\dd u,
  \qquad q\in[0,\alpha],
\]
and shows that its trace is the graph of an analytic function \(y_*:[0,1]\to\bbR\). With \(\lambda_*\defby\fe(C_*,s_*)\), the same proposition shows that \((\lambda_*,y_*)\) satisfies the Euler--Lagrange equation and the primal microscopic-length constraint.

The global minimality of this stationary solution was proved in Section~\ref{sec:VP:Minimizer}: every primal-admissible pair \((\lambda,\MP)\), after passing to its monomer-time representation, has action at least \(\calB(\eta_*)=\calA(\MP_*,\lambda_*)\). Hence \((\lambda_*,\MP_*)\) is a primal minimizer.

Proposition~\ref{prop:AC-uniqueness} identifies the equality case. It shows that any primal minimizer has \(\lambda=\lambda_*\), and that, after removing constant pieces, its trace is a nondecreasing absolutely continuous reparameterization of \(\eta_*\). Conversely, every such reparameterization is a primal minimizer. This proves uniqueness of the minimizing trace.

Finally, Lemma~\ref{lem:uniform-gap} gives the existence of \(\delta_*>0\) such that
\[
  \lambda_*+g y_*(x)\ge\lambdac+\delta_*
  \qquad\text{for all }x\in[0,1].
\]
All the assertions of Theorem~\ref{thm:SolutionPrimalVP} have been established.
\end{proof}

%%%%
\subsection{Equivalence of the primal and dual variational problems: proof of Theorem~\ref{thm:Main:EquVP}}
\label{sec:VP:PrimalDual}
%%%%

\begin{proof}[Proof of Theorem~\ref{thm:Main:EquVP}]
  Let \((\MP,\rho)\) be dual-admissible, and let \(\eta\) be its monomer-time representation, defined in the usual way by \(q(t)=\int_0^t\rho\). Then \(\eta\) is admissible for \(\calB\) and
  \[
    \calA^*(\MP,\rho)=\calB(\eta).
  \]
  Conversely, every admissible curve \(\eta:[0,\alpha]\to\bbR^2\) gives a dual-admissible pair, for instance by taking \(\rho(t)=\alpha\) and \(\MP(t)=\eta(\alpha t)\). Hence the dual problem is equivalent to minimizing \(\calB\), and therefore
  \[
    \inf\calA^* = \min\calB = \calB(\eta_*).
  \]

  We next compare this with the primal problem. Let \((\lambda,\MP)\) be primal-ad\-mis\-si\-ble, and define the associated density \(\rho(t) = \partial_\mu\icl_\mu(\dot\MP(t)) \vert_{\mu=\lambda+g\MP_2(t)}\).
  By the primal microscopic-length constraint, \(\int_0^1\rho=\alpha\), so \((\MP,\rho)\) is dual-ad\-mis\-si\-ble. By the equality case in Lemma~\ref{lem:perspective},
  \[
    \calA(\MP,\lambda) = \calA^*(\MP,\rho).
  \]
  Consequently,
  \[
    \inf\calA \geq \inf\calA^*.
  \]
  On the other hand, the stationary solution constructed above satisfies
  \[
    \calA(\MP_*,\lambda_*) = \calB(\eta_*) = \inf\calA^*,
  \]
  and therefore
  \[
    \inf\calA = \inf\calA^* = \calB(\eta_*).
  \]
  \noindent
  \textit{Primal to dual.}
  Let now \((\lambda,\MP)\) be a primal minimizer, and define \(\rho\) as above. Then
  \[
    \calA^*(\MP,\rho) = \calA(\MP,\lambda) = \inf\calA = \inf\calA^*,
  \]
  so \((\MP,\rho)\) is a dual minimizer.

  \noindent
  \textit{Dual to primal.}
  Conversely, let \((\MP,\rho)\) be a dual minimizer, and let \(\eta\) be its monomer-time representation. Then
  \[
    \calB(\eta)=\calA^*(\MP,\rho)=\inf\calA^*=\calB(\eta_*).
  \]
  By uniqueness of the minimizer of \(\calB\), \(\eta=\eta_*\). Thus, if \(q(t)=\int_0^t\rho(u)\,\dd u\), then \(\MP(t)=\eta_*(q(t))\). Set \(\lambda_*=\fe(C_*,s_*)\). By the stationary identity \(\lambda_*+g\eta_{*,2}(q)=\fe(h_*(q))\), we have \(\lambda_* + g\MP_2(t) = \fe(h_*(q(t))) > \lambdac\) for a.e.\ \(t\). Moreover, \(\dot\MP(t)=\rho(t)\eta_*'(q(t)) = \rho(t)\nabla\fe(h_*(q(t)))\) for a.e.\ \(t\). Hence, by the equality case in Lemma~\ref{lem:perspective},
  \[
    \rho(t)
    =
    \partial_\mu\icl_\mu(\dot\MP(t))
    \vert_{\mu=\lambda_*+g\MP_2(t)}
  \]
  for a.e.\ \(t\). In particular, since \(\int_0^1\rho(t)\,\dd t=\alpha\), the pair \((\lambda_*,\MP)\) is primal-admissible. The same equality case also gives
  \[
    \calA(\MP,\lambda_*)=\calA^*(\MP,\rho)=\inf\calA.
  \]
  Thus \((\lambda_*,\MP)\) is a primal minimizer.
\end{proof}

%%%%
\subsection{Stability of the minimizer}
\label{sec:VP:Stability}
%%%%

The preceding subsections identified the minimizer and established uniqueness by exploiting the convex-dual structure of the problem. We now prove a (non-quantitative) stability result, using a standard compactness and lower-semicontinuity argument from the direct method in the calculus of variations; see, e.g., \cite{Dacorogna-2008}.
\begin{proposition}\label{prop:VP:Stability}
Let \(\frm_*=(\eta_*)_\sharp\Leb_{[0,\alpha]}\).
For every \(\epsilon>0\), there exists \(c_\epsilon>0\) such that the following statements hold.
\begin{enumerate}
  \item If \(\eta:[0,\alpha]\to\bbR^2\) is admissible for the monomer-time functional, then
  \[
    d_{\mathrm{BL}}\bigl(\eta_\sharp\Leb_{[0,\alpha]},\frm_*\bigr) \geq \epsilon \implies
    \calB(\eta)\geq \calB(\eta_*)+c_\epsilon.
  \]
  \item If \((\MP,\rho)\) is admissible for the dual problem, then
  \[
    d_{\mathrm{BL}} \bigl( \MP_\sharp(\rho(t)\,\dd t), \frm_* \bigr) \geq\epsilon \implies
    \calA^*(\MP,\rho) \geq \calA^*(\MP_*,\rho_*)+c_\epsilon.
  \]
  \item If \((\lambda,\MP)\) is primal-admissible and the associated monomer-time representation gives a monomer measure at
  \(d_{\mathrm{BL}}\)-distance at least \(\epsilon\) from \(\frm_*\), then
  \[
    \calA(\MP,\lambda) \geq \calA(\MP_*,\lambda_*)+c_\epsilon.
  \]
\end{enumerate}
\end{proposition}
\begin{proof}
\textit{1.} Suppose, by contradiction, that the statement is false. Then there exists a sequence of admissible curves
\((\eta_n)_{n\ge1}\) such that \(d_{\mathrm{BL}}\bigl((\eta_n)_\sharp\Leb_{[0,\alpha]},\frm_*\bigr) \geq \epsilon\),
while \(\calB(\eta_n)\downarrow\calB(\eta_*)\).

Since \(J(v)=+\infty\) whenever \(\normI{v}>1\), every curve with finite \(\calB\)-value satisfies \(\normI{\eta_n'(q)} \leq 1\) for a.e.\ \(q\in [0,\alpha]\).
Thus \((\eta_n)\) is equi-Lipschitz and uniformly bounded. By Arzelà--Ascoli, after passing to a subsequence, \(\eta_n\) converges uniformly to an admissible curve \(\eta\).

The functional \(\calB\) is lower semicontinuous under this convergence: the term \(\int_0^\alpha J(\eta'(q))\,\dd q\) is lower semicontinuous by convexity and lower semicontinuity of \(J\), whereas \(g\int_0^\alpha \eta_2(q)\,\dd q\) is continuous under uniform convergence. Hence \(\calB(\eta) \leq \liminf_{n\to\infty}\calB(\eta_n) = \calB(\eta_*)\).

By uniqueness of the minimizer of \(\calB\), we have \(\eta=\eta_*\). The uniform convergence \(\eta_n\to\eta_*\) implies
\[
  (\eta_n)_\sharp\Leb_{[0,\alpha]} \xrightarrow{n\to\infty} (\eta_*)_\sharp\Leb_{[0,\alpha]}
\]
in \(d_{\mathrm{BL}}\), contradicting the assumption. This proves the first statement.

\textit{2.} The dual statement follows by passing to the monomer-time representation of \((\MP,\rho)\). Indeed, if \(\eta\) is this representation, then
\[
  \MP_\sharp(\rho(t)\,\dd t) = \eta_\sharp\Leb_{[0,\alpha]}
  \quad\text{and}\quad
  \calA^*(\MP,\rho)=\calB(\eta).
\]
\textit{3.} The primal statement follows similarly from the associated dual pair and the perspective identity.
\end{proof}

%%%%
\subsection{An exactly solvable case: the simple random walk}
\label{sec:VP:SRW}
%%%%

In this section, we provide some explicit computations in the simplest model in the class: the simple random walk, corresponding to \(\Phi(\gamma)\equiv 0\). For the latter, it is possible to determine the relevant thermodynamic quantities, as well as the resulting geodesics.

\paragraph{Thermodynamic quantities.}

A direct computation yields
\[
  \Z{h}{n}
  = \sum_{\gamma:\, \abs{\gamma}=n} e^{\langle h,X(\gamma)\rangle}
  = \bigl(e^{h_1}+e^{-h_1}+e^{h_2}+e^{-h_2}\bigr)^n,
\]
which shows that the free energy is given by
\[
  \fe(h) = \log\bigl(2\cosh h_1+2\cosh h_2\bigr).
\]

\smallskip
The inverse correlation length can be computed directly from its definition as the exponential decay rate of \(\Z{\lambda}{x}\). This leads to the explicit representation
\[
  \icl_\lambda(x)
  =
  \sum_{i=1}^2 x_i\,\arsinh(t x_i),
\]
where \(t=t(x,\lambda)\) is the unique positive solution of
\[
  \sqrt{1+t^2 x_1^2}+\sqrt{1+t^2 x_2^2}=\tfrac12 e^\lambda.
\]
\begin{figure}[t]
  \centering
  \includegraphics[height=5cm]{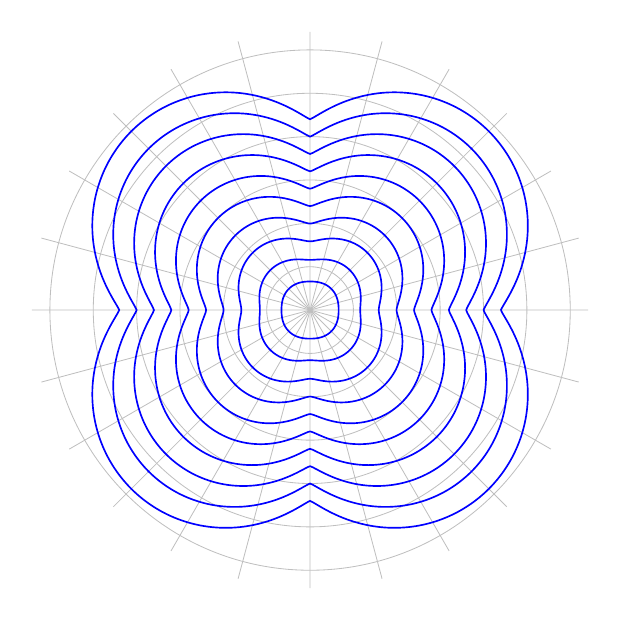}
  \hspace{1cm}
  \includegraphics[height=5cm]{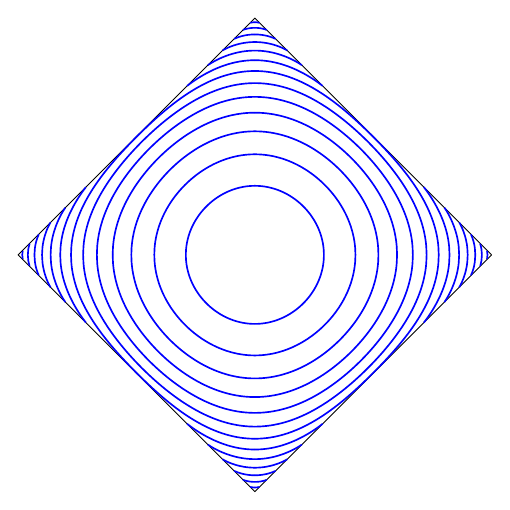}
  \caption{Left: Polar plot of \(\icl_\lambda(x)\) for \(x\in\bbS^1\) and \(\lambda\) taking values from \(1\) to \(10\). Right: Level lines of the rate function \(J\) on its effective domain (from \(-0.1\) to \(-1.3\) by steps of \(-0.1\)).}
  \label{fig:ICL+RateFct}
\end{figure}

\smallskip
Finally, the rate function \(J\), that is the Legendre transform of \(\fe\), is given as follows. For \(\normI{v} < 1\),
\[
  J(v) = \sum_{i=1}^2 v_i\,\arsinh(\tau v_i) - \log(2\tau),
\]
where \(\tau=\tau(v)>0\) is the unique solution of \(\sqrt{1+\tau^2v_1^2} + \sqrt{1+\tau^2v_2^2} = \tau\).

The rate function extends continuously to the boundary \(\normI v=1\), where
\[
  J(v)=\sum_{i=1}^2 \abs{v_i}\log \abs{v_i},
\]
with the convention \(0\log 0=0\), and \(J(v)=+\infty\) for \(\normI v>1\).

\smallskip
Figure~\ref{fig:LevelLines} shows some level lines of \(\fe\), while Figure~\ref{fig:ICL+RateFct} shows plots of the inverse correlation length and the rate function.

\paragraph{The geodesic.}

We now specialize the general parametric representation~\eqref{eq:parameteric-solution} to the simple random walk. Writing
\(s_0\) for the initial value of the dual variable and \(s_1=s_0+g\alpha\), we obtain
\[
  x(s) = \frac1g\int_{s_0}^s \partial_1\fe(C,u)\,\dd u,
  \qquad
  y(s) = \frac{\fe(C,s)-\fe(C,s_0)}{g},
  \qquad
  s\in[s_0,s_1].
\]
Since \(\partial_1\fe(C,s) = \sinh C/(\cosh C+\cosh s)\), this gives
\[
  x(s) = \frac{\sinh C}{g} \int_{s_0}^s \frac{\dd u}{\cosh C+\cosh u},
  \qquad
  y(s) = \frac1g \log \frac{\cosh C+\cosh s}{\cosh C+\cosh s_0}.
\]
The parameters \((C,s_0)\) are determined by \(x(s_1)=1\) and \(y(s_1)=a\):
\[
  1 = \frac{\sinh C}{g} \int_{s_0}^{s_0+g\alpha} \frac{\dd s}{\cosh C+\cosh s},
  \qquad
  e^{ga} = \frac{\cosh C+\cosh s_1}{\cosh C+\cosh s_0}.
\]
Using
\[
  \int \frac{\dd s}{\cosh C+\cosh s} = \frac{2}{\sinh C} \artanh\Bigl( \tanh\frac{C}{2}\tanh\frac{s}{2} \Bigr)
  \bydef \frac{2}{\sinh C}\Theta_C(s),
\]
we may write
\[
  x(s) = \frac2g \bigl[\Theta_C(s)-\Theta_C(s_0)\bigr].
\]

\paragraph{Symmetric case \(a=0\).}
When \(a=0\), symmetry implies \(s_0=-\frac{g\alpha}{2}\) and \(s_1=\frac{g\alpha}{2}\).
The horizontal constraint becomes
\[
  \tanh\frac{g}{4}
  =
  \tanh\frac{C}{2}\,
  \tanh\frac{\alpha g}{4}.
\]
Since \(g>0\), this equation has a unique solution \(C>0\) if and only if \(\alpha>1\). The curve is then given by
\begin{equation}\label{eq:geodesicSRW}
  x(s) = \frac2g \bigl[ \Theta_C(s) - \Theta_C(-g\alpha/2) \bigr],
  \qquad
  y(s) = \frac1g \log \frac{\cosh C+\cosh s}{\cosh C+\cosh(g\alpha/2)} ,
\end{equation}
with \(s\in[-g\alpha/2,g\alpha/2]\).

%%%%
\subsection{Behavior at small \texorpdfstring{\(g\)}{g}: proof of Proposition~\ref{pro:smallg}}
\label{sec:VP:Smallg}
%%%%

We work in the symmetric case \(a=0\). Let \(h_0=(C_0,0)\) be the
solution of the shooting system at \(g=0\). Thus \(\alpha\,\partial_1\fe(C_0,0)=1\) and \(\partial_2\fe(C_0,0)=0\).
The second identity follows from symmetry. For \(g>0\), symmetry of \(\fe(C,s)\) in the second coordinate and uniqueness in the shooting system imply \(s_*=-g\alpha/2\).
Indeed, the vertical shooting constraint becomes
\[
  \int_0^\alpha \partial_2\fe(C_*,s_*+gu)\,\dd u=0,
\]
and this integral vanishes when the interval \([s_*,s_*+g\alpha]\) is centered at \(0\).

We next expand the horizontal constraint. Since \(s_*=-g\alpha/2\),
\[
  1 = \int_0^\alpha \partial_1\fe(C_*,g(u-\alpha/2))\,\dd u.
\]
By evenness in the second variable, the first-order term in \(g\) vanishes. Hence \(C_*=C_0+\sfO(g^2)\).

We now use the monomer-time parameterization
\[
  x(q) = \int_0^q \partial_1\fe(C_*,g(u-\alpha/2))\,\dd u,
  \qquad
  y(q) = \int_0^q \partial_2\fe(C_*,g(u-\alpha/2))\,\dd u .
\]
Using \(C_*=C_0+\sfO(g^2)\), the symmetry of \(\fe\), and Taylor expansion at \((C_0,0)\), we obtain uniformly for \(q\in[0,\alpha]\),
\[
  x(q) = \frac{q}{\alpha} + \sfO(g^2)
\quad\text{and}\quad
  y(q) = g\,\partial_{22}\fe(C_0,0) \Bigl( \frac{q^2}{2}-\frac{\alpha q}{2} \Bigr) + \sfO(g^3).
\]
Inverting the first relation gives \(q(x)=\alpha x+\sfO(g^2)\), uniformly in \(x\in [0,1]\). Therefore
\[
  y_*(x)
  =
  \kappa_g(x^2-x)+\sfO(g^3),
\]
where
\[
  \kappa_g
  =
  g\,\frac{\alpha^2}{2}\,\partial_{22}\fe(C_0,0)
  +
  \sfO(g^3).
\]
The absence of a quadratic correction follows again from the symmetry in the second coordinate.
Since \(\Hess\fe(C_0,0)\) is positive definite, one has \(\partial_{22}\fe(C_0,0)>0\), and hence \(\kappa_g>0\) for all sufficiently small \(g\). Thus the polymer geodesic is, to leading order, an upward parabola.

It remains to compare this parabola with the classical symmetric catenary of the same apparent length. The apparent length of \(y_*\) is
\[
  L_g
  =
  \int_0^1 \sqrt{1+y_*'(x)^2}\,\dd x
  =
  1+\frac{\kappa_g^2}{6}+\sfO(g^4).
\]

The symmetric catenary connecting \((0,0)\) and \((1,0)\) can be written,
after centering at \(x=1/2\), in the form
\[
  y_{\mathrm{cat}}(x)
  =
  A\left[
    \cosh\left(\frac{x-1/2}{A}\right)
    -
    \cosh\left(\frac{1}{2A}\right)
  \right],
\]
with \(A\to\infty\) in the small-sag regime. Expanding for large \(A\),
\[
  y_{\mathrm{cat}}(x)
  =
  \frac{1}{2A}(x^2-x)
  +
  \sfO(A^{-3}),
\]
and its length is
\[
  L_{\mathrm{cat}}
  =
  1+\frac{1}{24A^2}+\sfO(A^{-4}).
\]
Choosing \(A\) so that \(L_{\mathrm{cat}}=L_g\) gives
\[
  \frac{1}{2A}
  =
  \kappa_g+\sfO(g^3).
\]
Consequently,
\[
  y_{\mathrm{cat}}(x)
  =
  \kappa_g(x^2-x)+\sfO(g^3),
\]
uniformly in \(x\in[0,1]\). Comparing this with the expansion of \(y_*\)
yields
\[
  \sup_{x\in[0,1]}
  \abs{y_*(x)-y_{\mathrm{cat}}(x)}
  \leq Cg^3,
\]
as claimed.
\myQED

%%%%
\subsection{Behavior at large \texorpdfstring{\(g\)}{g}: proof of Proposition~\ref{pro:largeg}}
\label{sec:VP:Largeg}
%%%%

Let us start by considering the limiting (\(g=\infty\)) variational problem.
\begin{lemma}\label{lem:VP:ginfinity}
Let \(a\in\bbR\) and \(\alpha>1+\abs{a}\). Among all absolutely continuous curves \(\eta:[0,\alpha]\to\bbR^2\) satisfying
\[
  \eta(0)=(0,0),\qquad
  \eta(\alpha)=(1,a),\qquad
  \normI{\eta'(q)} \leq 1 \quad\text{for a.e.\ }q,
\]
the functional
\[
  \eta\mapsto I(\eta) \defby \int_0^\alpha \eta_2(q)\,\dd q
\]
has a unique minimizing trace. This trace is the polygonal curve joining
\[
  (0,0),\quad (0,b_*),\quad (1,b_*),\quad (1,a),
  \qquad
  b_*=\frac{1+a-\alpha}{2}.
\]
\end{lemma}

\begin{proof}
Set \(q_- \defby -b_*\) and \(q_+\defby 1+q_-\). Since \(\alpha>1+\abs{a}\), \(0 < q_- < q_+ < \alpha\). Let \(\bar\eta\) be the curve which goes vertically from \((0,0)\) to \((0,b_*)\) during the time interval \([0,q_-]\), horizontally from \((0,b_*)\) to \((1,b_*)\) during \([q_-,q_+]\), and vertically from \((1,b_*)\) to \((1,a)\) during \([q_+,\alpha]\). Equivalently, \(\bar\eta_2(q)=\max\{-q,b_*,a-\alpha+q\}\).
This curve is admissible and satisfies \(\normI{\bar\eta'(q)}=1\) for a.e.\ \(q\in [0,\alpha]\).

Let now \(\eta\) be any admissible curve. The constraint \(\abs{\eta_1'(q)} + \abs{\eta_2'(q)}\leq 1\) implies that, for all \(q\in [0,\alpha]\),
\[
  \eta_2(q)
  =
  \int_0^q \eta_2'(s)\,\dd s
  \geq
  \int_0^q\bigl(\abs{\eta_1'(s)}-1\bigr)\,\dd s
  \geq -q,
\]
and
\[
  \eta_2(q)
  =
  a-\int_q^\alpha \eta_2'(s)\,\dd s
  \geq
  a-\int_q^\alpha\bigl(1-\abs{\eta_1'(s)}\bigr)\,\dd s
  \geq
  a-\alpha+q.
\]
Adding the two sharper estimates gives
\[
  2\eta_2(q) \geq a-\alpha+\int_0^\alpha\abs{\eta_1'(s)}\,\dd s \geq a-\alpha+1 = 2b_*,
\]
since \(\int_0^\alpha\abs{\eta_1'(s)}\,\dd s \geq \int_0^\alpha\eta_1'(s)\,\dd s = 1\).
Therefore, \(\eta_2(q) \geq \max\{-q,b_*,a-\alpha+q\} = \bar\eta_2(q)\) for all \(q\in[0,\alpha]\).
It follows immediately that
\[
  I(\eta) = \int_0^\alpha \eta_2(q)\,\dd q \geq \int_0^\alpha \bar\eta_2(q)\,\dd q = I(\bar\eta).
\]
It remains only to identify the equality case.
If equality holds in the last display, then the continuous nonnegative function \(\eta_2-\bar\eta_2\) vanishes identically. The constraint then forces \(\eta_1'=0\) on the two vertical pieces of \(\bar\eta_2\), while on the plateau the interval has length \(1\) and the total horizontal displacement is \(1\); hence \(\eta_1'=1\) there. We conclude that \(\eta=\bar\eta\).
\end{proof}

\begin{proof}[Proof of Proposition~\ref{pro:largeg}]
By minimality and by the lower bound \(J \geq -\lambdac\),
\[
  \forall g>0,\qquad
  gI(\eta_g) - \lambdac\alpha
  \leq
  \calB_g(\eta_g)
  \leq
  \calB_g(\bar\eta)
  =
  gI(\bar\eta) + \sfO(1).
\]
Hence
\begin{equation}\label{eq:ineq_largeg}
  \forall g>0,\qquad
  I(\eta_g) \leq I(\bar\eta)+\sfO(1/g).
\end{equation}
Now let \(g_n\to\infty\). Since finite action implies \(\normI{\eta_{g_n}'}\leq 1\) a.e., Arzelà--Ascoli gives a uniformly
convergent subsequence with limit \(\eta_\infty\). The inequality~\eqref{eq:ineq_largeg} passes to the limit and gives \(I(\eta_\infty)\leq I(\bar\eta)\). By Lemma~\ref{lem:VP:ginfinity}, \(\eta_\infty=\bar\eta\). Therefore the whole family converges uniformly to \(\bar\eta\).
\end{proof}

%%%%%%%%%%%%%%%%%%%%%%%
\section{Proofs of the main results}
\label{sec:ProofMain}
%%%%%%%%%%%%%%%%%%%%%%%

In this section, we prove our main results: Theorems~\ref{thm:FreeEnergy} and~\ref{thm:Concentration}, and Corollary~\ref{cor:TubeConcentration}. The section is structured as follows. We first prove Theorem~\ref{thm:FreeEnergy} in Section~\ref{sec:ProofMain:FreeEnergy}; the proof has two parts: a lower bound derived in Section~\ref{sec:ProofMain:FreeEnergy:LowerBound} and a matching upper bound established in Section~\ref{sec:ProofMain:FreeEnergy:UpperBound}. We then use this information, and the stability of the variational problem established in Proposition~\ref{prop:VP:Stability}, to prove Theorem~\ref{thm:Concentration} and Corollary~\ref{cor:TubeConcentration} in Section~\ref{sec:ProofMain:Concentration}.

%%%%
\subsection{Proof of Theorem~\ref{thm:FreeEnergy}}
\label{sec:ProofMain:FreeEnergy}
%%%%

%
\subsubsection{Lower bound on the free energy}\label{sec:ProofMain:FreeEnergy:LowerBound}

Before turning to the actual proof, let us note the following consequence of the repulsivity of the self-interaction and of the normalization \(\phi(1)=0\). Suppose that \(\gamma':0\to x\) and \(\gamma'':x\to z\), and that \(\gamma'\) visits \(x\) only at its final time, then
\begin{align}
  \Phi(\gamma'\concat\gamma'')
  &= \phi(\ell_x(\gamma'\concat\gamma'')) + \sum_{y\neq x} \phi(\ell_y(\gamma'\concat\gamma'')) \notag \\
  &= \phi(\ell_x(\gamma') +\ell_x(\gamma'') - 1) + \sum_{y\neq x} \phi(\ell_y(\gamma') + \ell_y(\gamma'')) \notag \\
  &\geq \phi(\ell_x(\gamma'')) + \sum_{y\neq x} \bigl(\phi(\ell_y(\gamma')) + \phi(\ell_y(\gamma''))\bigr) \notag \\
  &= \Phi(\gamma') + \Phi(\gamma''), \label{eq:IneqConcat}
\end{align}
where we used \(\ell_x(\gamma')=1\) and~\eqref{eq:condphi} for the third line, and \(\phi(1)=0\) for the last one.

\medskip
We are now ready to start the proof. For \(z\in\bbZ^2\), set
\[
    V_N(z) \defby \frac{g}{N}z_2,\qquad g>0.
\]
We introduce the partition functions in the presence of this field:
\begin{gather*}
    \Z{g}{n,x} \defby \sum_{\substack{\gamma: 0 \to x \\ \abs{\gamma}=n}} \W(\gamma) \exp\Bigl( - \sum_{k=1}^n V_N(\gamma_k) \Bigr), \\
    \Z{g;\mu}{x,y} \defby \sum_{\gamma: x \to y} \W(\gamma) \exp\Bigl( - \sum_{k=1}^{\abs{\gamma}} (\mu + V_N(\gamma_k)) \Bigr).
\end{gather*}

Let \((\lambda_*,\MP_*)\) be the unique solution of the primal variational problem constructed in Section~\ref{sec:VP}.
The goal of this section is to prove that
\begin{equation}\label{eq:ProofMain:LBFreeEnergy}
  \liminf_{N\to\infty} -\frac1N\log \Z{g}{L_N,A_N} \geq \calA(\MP_*,\lambda_*).
\end{equation}

\paragraph{Coarse-graining.}
Let us fix some \(\delta>0\).
Fix \(0<\zeta<1\) and set \(R_N\defby \lfloor N^\zeta\rfloor.\)
For a nearest-neighbor path \(\gamma\) contributing to \(\Z{g}{L_N,A_N}\), define
\[
    \tau_0\defby 0,
\]
and, for \(k\geq0\),
\[
    \tau_{k+1} \defby \inf\setof{t>\tau_k}{\normsup{\gamma_t-\gamma_{\tau_k}}\geq R_N}\wedge L_N.
\]
Set also, for all \(k\geq 0\),
\[
  v_k\defby \gamma_{\tau_k}.
\]
Let \(M=M(\gamma)\) be the smallest integer such that \(\tau_M=L_N\), and define the skeleton
\[
    \Skel{\gamma} \defby (\bfv,\bftau),
\]
where \(\bfv = (v_0,\dots,v_M)\) and \(\bftau = (\tau_0,\dots,\tau_M)\).

By construction \(v_0=0\), \(v_M=A_N\), and, for \(k<M-1\), \(\normsup{v_{k+1}-v_k}=R_N\), while the last increment satisfies \(\normsup{v_M-v_{M-1}}\leq R_N\).
Moreover, \((M-1)R_N\leq L_N\), which implies that
\[
  M\leq M_N\defby \left\lfloor\frac{L_N}{R_N}\right\rfloor+1 = \sfO(N^{1-\zeta}).
\]
Let \(\calV_{N,M}\) be the set of skeletons with \(M\) segments arising from paths of length \(L_N\) from \(0\) to \(A_N\), and put
\[
    \calV_N\defby \bigcup_{1\leq M\leq M_N}\calV_{N,M}.
\]
We partition the paths according to their skeletons, obtaining
\[
  \Z{g}{L_N,A_N}
  = \sum_{M=1}^{M_N} \sum_{(\bfv,\bftau)\in\calV_{N,M}} \sum_{\substack{\gamma:\,0\to A_N\\\Skel{\gamma}=(\bfv,\bftau)}} \W(\gamma)
  \exp\Bigl\{
      -\sum_{j=1}^{L_N} V_N(\gamma_j)
  \Bigr\}.
\]

\paragraph{Upper bound on the partition function.}
We fix a skeleton \((\bfv,\bftau)\in\calV_{N,M}\), and write
\[
  \Delta v_k\defby v_{k}-v_{k-1}
  \qquad\text{and}\qquad
  n_k \defby \tau_{k}-\tau_{k-1}.
\]
For every \(j\in\{\tau_{k-1}+1,\dots,\tau_k\}\), one has by construction \(\normsup{\gamma_j-v_{k-1}}\le R_N\).
Therefore \((\gamma_j)_2\geq (v_{k-1})_2-R_N\), which implies that
\[
    V_N(\gamma_j)
    =
    \frac{g}{N}(\gamma_j)_2
    \geq
    \frac{g}{N}(v_{k-1})_2 -\epsilon_N,
\]
where \(\epsilon_N\defby \frac{gR_N}{N}\).
For \(k<M\), the \(k\)-th skeleton segment visits its endpoint \(v_k\) only at its final time, by the definition of \(\tau_k\). Applying~\eqref{eq:IneqConcat} iteratively from right to left therefore yields
\begin{multline}\label{eq:UB_segment_product}
  \sum_{\substack{\gamma:\,0\to A_N\\\Skel{\gamma}=(\bfv,\bftau)}} \W(\gamma) \exp\Bigl\{ -\sum_{j=1}^{L_N} V_N(\gamma_j) \Bigr\} \notag\\
  \leq
  \prod_{k=1}^{M}
      \Bigl(
          \sum_{\substack{\eta:\,0\to \Delta v_k\\\abs{\eta} = n_k}} \W(\eta) e^{-(g(v_{k-1})_2/N - \epsilon_N) \abs{\eta}}
      \Bigr)\\
  = \prod_{k=1}^{M} \exp\Bigl\{- \frac{n_k}{N}g(v_{k-1})_2 + \epsilon_N n_k \Bigr\} \Z{}{n_k,\Delta v_k}.
\end{multline}
We apply Lemma~\ref{lem:logAsymptotics:UpperZ} to every segment with \(n_k\geq n_\delta\), obtaining
\[
  \Z{}{n_k,\Delta v_k} \leq \exp\Bigl\{-n_k J\bigl( \frac{\Delta v_k}{n_k} \bigr)  + \delta n_k \Bigr\}.
\]
Notice that all segments except the last one automatically satisfy this condition once \(N\) is sufficiently large, since \(n_k\geq R_N\). If the last segment has \(n_M<n_\delta\), we use the crude bound \(\Z{}{n_M,\Delta v_M}\le \Z{0}{n_M}\le 4^{n_\delta}\), which is \(\sfO(1)\). Otherwise the lemma applies to it as well. Therefore,
\begin{align*}
  -\frac1N \log& \biggl[ \sum_{\substack{\gamma:\,0\to A_N\\\Skel{\gamma}=(\bfv,\bftau)}} \W(\gamma) \exp\Bigl\{ -\sum_{j=1}^{L_N} V_N(\gamma_j) \Bigr\} \biggr] \\
  &\geq
  \sum_{k=1}^{M} \Bigl( \frac{n_k}{N}g\frac{(v_{k-1})_2}{N} + \frac{n_k}{N} J\bigl( \frac{\Delta v_k/N}{n_k/N} \bigr) \Bigr) - \frac{L_N}{N}(\epsilon_N+\delta) - \sfO(1/N) \\
  &\bydef A^*(\bfv,\bftau) - \frac{L_N}{N}(\epsilon_N+\delta) - \sfO(1/N).
\end{align*}
Let us turn to the entropy associated to the skeletons, that is, let us bound \(\abs{\calV_{N}}\). First, the number of possible \(\bfv\) is easily bounded: given \(v_k\), there are at most \((2R_N+1)^2 \leq 9R_N^2\) possible choices for \(v_{k+1}\). This shows that the total number of possible \(\bfv\) is bounded above by \((9R_N^2)^M = e^{\sfO(N^{1-\zeta}\log N)}\). The number of possible choices for \(\bftau\) is bounded above by
\[
  \binom{L_N + M}{M} \leq \biggl(\frac{e(L_N+M)}{M}\biggr)^M = e^{\sfO(N^{1-\zeta}\log N)}.
\]
After summing over \(M\), we conclude that \(\abs{\calV_{N}} \leq e^{\sfO(N^{1-\zeta}\log N)}\). We thus have
\begin{multline*}
	-\frac1N \log \Z{g}{L_N,A_N}
	\geq \inf_{(\bfv,\bftau)\in\calV_{N}} A^*(\bfv,\bftau) - \frac{L_N}{N}(\epsilon_N+\delta) - \sfO(N^{-\zeta}\log N).
\end{multline*}

\paragraph{Limiting behavior.}
Let \(\alpha_N\defby L_N/N\).
For \(1\leq k\leq M\), set
\[
  \rho_k\defby \frac{n_k}{N},
  \qquad
  \Delta\MP_k\defby \frac{\Delta v_k}{N},
  \qquad
  \MP_k\defby \frac{v_k}{N},
  \qquad
  t_k\defby \frac{\tau_k}{L_N}.
\]
Given a skeleton \((\bfv,\bftau)\), define \(\MP_{\bfv,\bftau}:[0,1]\to\bbR^2\) by linear interpolation:
for \(t\in[t_{k-1},t_k]\),
\begin{equation}\label{eq:skel_path}
  \MP_{\bfv,\bftau}(t) \defby \MP_{k-1} + (t-t_{k-1})\frac{L_N}{n_k}\Delta\MP_k .
\end{equation}
Define also, for \(t\in[t_{k-1},t_k]\),
\begin{equation}\label{eq:skel_density}
  \rho_{\bfv,\bftau}(t) \defby \frac{\rho_k}{t_k-t_{k-1}} = \frac{L_N}{N} = \alpha_N.
\end{equation}
Then
\[
  \int_0^1\rho_{\bfv,\bftau}(t)\,\dd t=\alpha_N \xrightarrow{N\to\infty}\alpha,
\]
and
\[
  \forall t\in(t_{k-1},t_k),\qquad
  \frac{\MP_{\bfv,\bftau}'(t)}{\rho_{\bfv,\bftau}(t)} = \frac{\Delta v_k}{n_k}.
\]
Consequently,
\[
  \sum_{k=1}^M \frac{n_k}{N} J\Bigl( \frac{\Delta v_k}{n_k} \Bigr)
  =
  \int_0^1 \rho_{\bfv,\bftau}(t) J\Bigl( \frac{\MP_{\bfv,\bftau}'(t)}{\rho_{\bfv,\bftau}(t)} \Bigr)\,\dd t,
\]
which shows that the kinetic terms agree exactly. Moreover, on each interval \([t_{k-1},t_k]\), the second coordinate of \(\MP_{\bfv,\bftau}\) differs from \((v_{k-1})_2/N\) by at most \(R_N/N\). Therefore the potential terms differ by at most \(C R_N/N=\sfo_N(1)\), uniformly in the skeleton. We conclude that
\[
  A^*(\bfv,\bftau) \geq \calA^*(\MP_{\bfv,\bftau},\rho_{\bfv,\bftau}) - \sfo_N(1),
\]
where
\[
  \int_0^1\rho_{\bfv,\bftau}(t)\,\dd t = \alpha_N\defby L_N/N.
\]
Let us denote by \(\Z{g}{L_N,A_N}[\bfv,\bftau]\) the contribution to \(\Z{g}{L_N,A_N}\) due to paths that have skeleton \((\bfv,\bftau)\). We have just proved the following.
\begin{lemma}\label{lem:ProofMain:SkeletonUpperBound}
With \(\delta>0\) fixed as above, uniformly in \((\bfv,\bftau)\in\calV_N\),
\[
  \Z{g}{L_N,A_N}[\bfv,\bftau]
  \leq
  \exp\left\{
    -N\calA^*(\MP_{\bfv,\bftau},\rho_{\bfv,\bftau})
    +N\alpha_N(\epsilon_N+\delta)
    +\sfo(N)
  \right\}.
\]
\end{lemma}
Now, taking the infimum over skeletons and then enlarging the class of competitors gives
\[
  \inf_{\calV_N} A^*
  \ge
  \inf_{\substack{\MP,\rho\\ \int\rho=\alpha_N}}
  \calA^*(\MP,\rho)-\sfo_N(1).
\]
Since \(L_N/N\to\alpha\), the mass constraint in the dual problem tends to the desired one. The minimal value of the dual problem is continuous in this mass parameter: indeed, by Proposition~\ref{prop:shooting-solution} the minimizer is obtained from the unique solution of the shooting system, and this solution depends continuously, in fact analytically, on \(\alpha\) by the implicit function theorem. Thus the infimum with constraint \(\int\rho=\alpha_N\) converges to the infimum with constraint \(\int\rho=\alpha\).

We have thus proved that
\[
\liminf_{N\to\infty} -\frac1N\log \Z{g}{L_N,A_N} \geq \inf_{\substack{\MP,\rho\\ \int\rho=\alpha}} \calA^*(\MP,\rho) - \alpha\delta = \calA(\MP_*,\lambda_*) - \alpha\delta,
\]
where the last identity follows from Theorem~\ref{thm:Main:EquVP}.
\(\delta\) being arbitrary, the desired bound~\eqref{eq:ProofMain:LBFreeEnergy} is proved.

%%%%
\subsubsection{Upper bound on the free energy}
\label{sec:ProofMain:FreeEnergy:UpperBound}
%%%%

\paragraph{Coarse-graining.}
Let \((\lambda_*,\MP_*)\) be the primal minimizer, written as \(\MP_*(x)=(x,y_*(x))\), and set \(\mu_*(x)=\lambda_*+g y_*(x)\). Define \(\rho_*(x) = \partial_\mu \icl_\mu(\dot\MP_*(x)) \vert_{\mu=\mu_*(x)}\).
Then \(\int_0^1\rho_*=\alpha\).
Let \(h_*(x)\in\partial \K_{\mu_*(x)}\) be the dual point to \(\MP_*'(x)\), so that \(\MP_*'(x)=\rho_*(x)\nabla\fe(h_*(x))\).

Fix \(M\geq 1\). Let
\[
  x_k\defby  \frac{k}{M},
  \qquad 0\leq k\leq M.
\]
For \(1\leq k\leq M-1\), set \(u_{k,N}\defby [N\MP_*(x_k)],\) \(u_{0,N}\defby 0,\) and \(u_{M,N}\defby A_N.\)
Let
\[
  \Delta u_{k,N}\defby u_{k,N}-u_{k-1,N}, \qquad
  \ell_k\defby \int_{x_{k-1}}^{x_k}\rho_*(s) \dd s,
\]
so that \(\sum_{k=1}^M \ell_k=\alpha\) and
\[
  \lim_{N \rightarrow \infty} \frac{\Delta u_{k,N}}{N} = \MP_*(x_k)-\MP_*(x_{k-1}).
\]

\begin{figure}[t]
  \centering
  \begin{tikzpicture}[scale=10]
    \def\M{10}
    \def\ymin{-0.55}

    \foreach \k in {0, ..., \M} {
        \pgfmathsetmacro{\xval}{\k/\M}
        \pgfmathsetmacro{\yval}{2*\xval*(\xval-1)}
        \draw[Gray!50!white, thin] (\xval, 0) -- (\xval, \ymin);% node[pos=1, below, black] {\(t_\k\)};
        \coordinate (P-\k) at (\xval, \yval);
    }

    \pgfmathtruncatemacro{\MminusOne}{\M-1}
    \foreach \m in {0, ..., \MminusOne} {
        \pgfmathtruncatemacro{\n}{\m+1}
        \pgfmathsetmacro{\Ax}{\m/\M}
        \pgfmathsetmacro{\Ay}{2*\Ax*(\Ax-1)}
        \pgfmathsetmacro{\Cx}{\n/\M}
        \pgfmathsetmacro{\Cy}{2*\Cx*(\Cx-1)}
        \pgfmathsetmacro{\dx}{\Cx - \Ax}
        \pgfmathsetmacro{\dy}{\Cy - \Ay}
        \pgfmathsetmacro{\theta}{atan(\dy/\dx)}
        \pgfmathsetmacro{\alpha}{\theta/2 - 45}
        \pgfmathsetmacro{\gamma}{\theta - \alpha}
        \pgfmathsetmacro{\L}{sqrt(\dx*\dx + \dy*\dy)}
        \pgfmathsetmacro{\LOne}{\L * cos(\gamma)}
        \pgfmathsetmacro{\LTwo}{\L * sin(\gamma)}
        \pgfmathsetmacro{\Bx}{\Ax + \LOne * cos(\alpha)}
        \pgfmathsetmacro{\By}{\Ay + \LOne * sin(\alpha)}
        \pgfmathsetmacro{\Dx}{\Ax + \LTwo * cos(\alpha + 90)}
        \pgfmathsetmacro{\Dy}{\Ay + \LTwo * sin(\alpha + 90)}
        \filldraw[black, thin, fill=yellow!30!white] (P-\m) -- (\Bx, \By) -- (P-\n) -- (\Dx, \Dy) -- cycle;
        \coordinate (B-\m) at (\Bx, \By);
        \coordinate (D-\m) at (\Dx, \Dy);
    }

    \draw[blue, thick, domain=0:1, samples=100] plot (\x, {2*\x*(\x-1)});
    \foreach \k in {0, ..., \M} {
        \filldraw[black, thin, fill=white] (P-\k) circle (.15pt);
    }
    \node[above=1.2] at (P-0) {\(u_{0,N}\)};
    \node[above=1.2] at (P-\M) {\(u_{M,N}\)};
    \node[label={[label distance=0pt]87:\(u_{1,N}\)}] at (P-1) {};
    \node[label={[label distance=0pt]80:\(u_{2,N}\)}] at (P-2) {};
    \draw[<->, thin] (0.3,0) -- (0.4,0) node[pos=.5, above] {\(N/M\)};
  \end{tikzpicture}
  \caption{The construction in Section~\ref{sec:ProofMain:FreeEnergy:UpperBound}. The blue curve is the minimizer of the variational problem. The partition function is lower bounded by summing only over paths that are concatenations of pieces contained inside successive diamonds; the length of each sub-path is also fixed to the value dictated by the target functional.}
\end{figure}
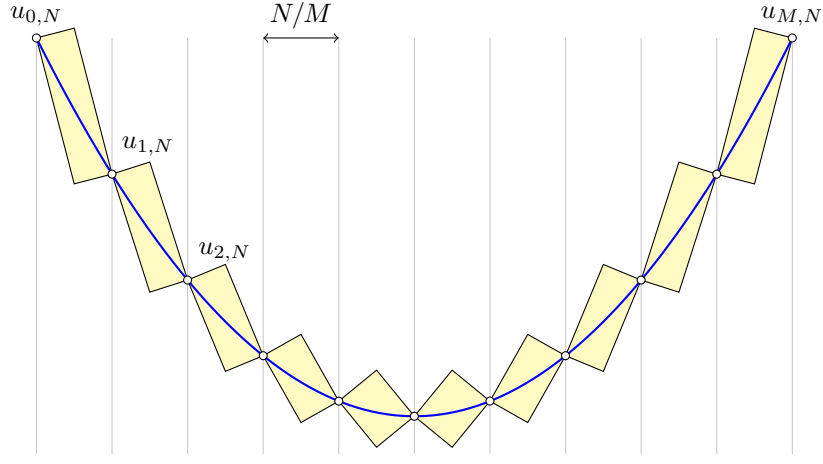

We use only admissible values of \(N\), namely those for which \(L_N-\normI{A_N}\) is nonnegative and even. For all sufficiently large \(N\), choose integers \(n_{k,N}\) such that
\[
    \sum_{k=1}^M n_{k,N}=L_N, \qquad n_{k,N}\equiv \normI{\Delta u_{k,N}}\pmod 2, \qquad \lim_{N\to\infty} \frac{n_{k,N}}{N}=\ell_k.
\]
Such a choice is possible because the parity constraints are compatible:
\[
  \sum_{k=1}^M \normI{\Delta u_{k,N}} \equiv \normI{A_N} \equiv L_N \pmod 2.
\]
Starting from integers close to \(N\ell_k\) with the prescribed parities, one adjusts finitely many of them by multiples of \(2\) to make the total sum equal to \(L_N\). Since \(M\) is fixed, this does not affect the limits \(n_{k,N}/N\to\ell_k\).

For each \(k\), consider paths from \(u_{k-1,N}\) to \(u_{k,N}\) obtained by translating paths contributing to \(\D{}{n_{k,N},\Delta u_{k,N}}\).
By Proposition~\ref{prop:shooting-solution}, \(y_*\) is analytic on \([0,1]\), and therefore the minimizing graph has bounded slope. The latter observation is used here to ensure that the diamonds are nondegenerate at the lattice scale. Since the directions of the increments \(\Delta u_{k,N}\) stay uniformly away from the vertical direction, our choice of aperture implies that the associated forward cone contains \(e_1\) with a uniform angular margin. Thus each diamond contains a nearest-neighbor path between its endpoints. The parity condition and the interior-velocity condition stated below then ensure that paths of the prescribed lengths \(n_{k,N}\) are available and are covered by the diamond-confined asymptotics of Lemma~\ref{lem:LogAsymptotics:LowerD}.

\paragraph{Lower bound on the partition function.}
For \(z = u_{k-1,N}+\eta_j\), where \(\eta\) is any such path from \(u_{k-1,N}\) to \(u_{k,N}\), one has
\[
  z_2 \leq (u_{k-1,N})_2 + n_{k,N}.
\]
Thus
\[
  V_N(z) = \frac{g}{N}z_2 \leq g y_*(x_{k-1}) + \varepsilon_{N,M},
\]
where
\[
  \varepsilon_{N,M}\defby g\max_{1\leq k\leq M}\abs[\Big]{\frac{(u_{k-1,N})_2}{N} - y_*(x_{k-1})} + g\max_{1\leq k\leq M}\frac{n_{k,N}}{N},
\]
which satisfies \(\limsup_{N\to\infty}\varepsilon_{N,M}\leq C/M\) for every \(M\).

The interiors of the successive diamond-confined pieces lie in disjoint diamonds, and the endpoints of such pieces are not revisited. Hence concatenation creates no additional intersections, except for the identification of consecutive endpoints. Since \(\phi(1)=0\), the
self-interaction energy is thus additive over the concatenated pieces. Hence, restricting the partition function to the concatenations described above gives
\[
  \Z{g}{L_N,A_N}\geq
  \prod_{k=1}^M
  \bigl[ \exp\{-n_{k,N}g y_*(x_{k-1})-n_{k,N}\varepsilon_{N,M}\}\D{}{n_{k,N},\Delta u_{k,N}} \bigr].
\]

\paragraph{Limiting behavior.}
For fixed \(M\), the velocities \(\Delta u_{k,N}/n_{k,N}\) remain, for all large \(N\), in a compact subset of the interior of the effective velocity
domain. Indeed,
\[
  \frac{\Delta u_{k,N}}{n_{k,N}}
  \xrightarrow{N\to\infty}
  \frac{\MP_*(x_k)-\MP_*(x_{k-1})}{\ell_k}
  =
  \frac{
    \int_{x_{k-1}}^{x_k}
    \nabla\fe(h_*(s))\,\rho_*(s)\,\dd s
  }{
    \int_{x_{k-1}}^{x_k}\rho_*(s)\,\dd s
  },
\]
where we used \(\MP_*'(s)=\rho_*(s)\nabla\fe(h_*(s))\).
Since \(h_*(s)\) ranges over a compact subset of \(\bbR^2\setminus\{0\}\), the vectors \(\nabla\fe(h_*(s))\) form a compact subset of the interior of the unit \(\ell^1\)-ball. Their weighted averages are therefore uniformly separated from the boundary \(\normI v=1\), and the same holds for \(\Delta u_{k,N}/n_{k,N}\) for all large \(N\). By Lemma~\ref{lem:LogAsymptotics:LowerD}, applied for fixed \(M\) to each compatible pair \((n_{k,N},\Delta u_{k,N})\),
\begin{multline*}
  \frac1N\log \Z{g}{L_N,A_N}
  \geq
  - \frac1N \sum_{k=1}^M n_{k,N} J \Bigl( \frac{\Delta u_{k,N}}{n_{k,N}} \Bigr) - \frac1N \sum_{k=1}^M n_{k,N}g y_*(x_{k-1}) \\
  - \varepsilon_{N,M}\frac{L_N}{N} + \sfo_N(1),
\end{multline*}
where \(\sfo_N(1)\) is for fixed \(M\). Letting \(N\to\infty\) yields
\begin{multline}\label{eq:LB_fini}
  \liminf_{N\to\infty} \frac1N \log \Z{g}{L_N,A_N}
  \geq
  - \sum_{k=1}^M \ell_k J\Bigl( \frac{\MP_*(x_k)-\MP_*(x_{k-1})}{\ell_k} \Bigr) \\
  - \sum_{k=1}^M \ell_k g y_*(x_{k-1}) - \frac{C}{M} .
\end{multline}
Now,
\[
  \lim_{M\to\infty} \sum_{k=1}^M \ell_k J\Bigl( \frac{\MP_*(x_k)-\MP_*(x_{k-1})}{\ell_k} \Bigr)
  = \int_0^1\rho_*(x) J\Bigl( \frac{\MP_*'(x)}{\rho_*(x)} \Bigr)\, \dd x
\]
and
\[
  \lim_{M\to\infty} \sum_{k=1}^M \ell_k g y_*(x_{k-1})
  =
  \int_0^1 \rho_*(x)g y_*(x)\, \dd x.
\]
The convergence of the first sum follows because \(\rho_*\), \(\MP_*'\), and \(J\) are continuous on the relevant compact subset of the interior of the effective domain, and
\[
  \frac{\MP_*(x_k)-\MP_*(x_{k-1})}{\ell_k}
\]
is the \(\rho_*\)-average of \(\MP_*'(x)/\rho_*(x)\) over \([x_{k-1},x_k]\). The second convergence is an ordinary weighted Riemann sum.

Combining this with~\eqref{eq:LB_fini} yields
\begin{align*}
  \liminf_{N\to\infty} \frac1N \log \Z{g}{L_N,A_N}
  &\geq
  -\int_0^1 \rho_*(x) \left[ J\Bigl(\frac{\MP_*'(x)}{\rho_*(x)}\Bigr) + g y_*(x) \right]\,\dd x \\
  &=
  -\calA^*(\MP_*,\rho_*) \\
  &=
  -\calA(\MP_*,\lambda_*),
\end{align*}
where the last equality follows from Theorem~\ref{thm:Main:EquVP}.
This establishes the upper bound on the free energy,
\begin{equation}\label{eq:ProofMain:UBFreeEnergy}
    \limsup_{N\to\infty} -\frac1N \log \Z{g}{L_N,A_N}
    \leq
    \calA(\MP_*,\lambda_*).
\end{equation}
Combining this bound with~\eqref{eq:ProofMain:LBFreeEnergy} concludes the proof of Theorem~\ref{thm:FreeEnergy}.\myQED

%%%%
\subsection{Proofs of Theorem~\ref{thm:Concentration} and Corollary~\ref{cor:TubeConcentration}}
\label{sec:ProofMain:Concentration}
%%%%

%
\subsubsection{Concentration: proof of Theorem~\ref{thm:Concentration}}
Let \(m_*\defby \calA^*(\MP_*,\rho_*) = \calA(\MP_*,\lambda_*)\).
Let \((\bfv,\bftau)\in\calV_N\) be the skeleton associated with a path \(\gamma\) and recall the definition of \(\MP_{\bfv,\bftau}\) and \(\rho_{\bfv,\bftau}\) in~\eqref{eq:skel_path} and~\eqref{eq:skel_density}.
We also associate with the skeleton the finite measure
\[
  \frm_{\bfv,\bftau} \defby (\MP_{\bfv,\bftau})_\sharp \bigl(\rho_{\bfv,\bftau}(t)\,\dd t\bigr).
\]
\begin{lemma}\label{lem:SkeletonMeasureApprox}
Uniformly over paths \(\gamma\) and their associated skeletons \((\bfv,\bftau)\in\calV_N\),
\[
  d_{\mathrm{BL}} ( \frm_\gamma^N, \frm_{\bfv,\bftau} ) \xrightarrow{N\to\infty} 0.
\]
\end{lemma}
\begin{proof}
Let \(\varphi\) be such that \(\normsup{\varphi}\leq 1\) and \(\mathrm{Lip}(\varphi)\leq 1\). We compare the two measures segment by segment. On the \(k\)-th segment, both the microscopic points \(\gamma_j/N\), \(j=\tau_{k-1}+1,\ldots,\tau_k\), and the interpolated macroscopic segment \(\MP_{\bfv,\bftau}([t_{k-1},t_k])\) lie within distance \(C R_N/N\) of \(v_{k-1}/N\). Therefore
\[
  \abs[\biggl]{
    \frac1N \sum_{j=\tau_{k-1}+1}^{\tau_k} \varphi(\gamma_j/N)
    - \int_{t_{k-1}}^{t_k} \varphi(\MP_{\bfv,\bftau}(t)) \rho_{\bfv,\bftau}(t)\,\dd t
  }
  \leq
  C\frac{n_k}{N}\frac{R_N}{N}.
\]
Summing over \(k\) gives
\[
  C\frac{R_N}{N}\sum_{k=1}^M\frac{n_k}{N} \leq C\alpha_N\frac{R_N}{N} .
\]
Since \(R_N/N\xrightarrow{N\to\infty} 0\), the claim follows after taking the
supremum over admissible test functions \(\varphi\).
\end{proof}
Fix \(\epsilon>0\). By Lemma~\ref{lem:SkeletonMeasureApprox}, for all large \(N\), if \(d_{\mathrm{BL}}(\frm_\gamma^N,\frm_*)>\epsilon\), then the associated skeleton satisfies \(d_{\mathrm{BL}}(\frm_{\bfv,\bftau},\frm_*)>\epsilon/2\).
By the measure stability estimate of Proposition~\ref{prop:VP:Stability}, applied uniformly for \(\alpha_N=L_N/N\) close to \(\alpha\), it follows that
\[
  \calA^*(\MP_{\bfv,\bftau},\rho_{\bfv,\bftau}) \geq m_* + \tfrac12c_{\epsilon/2}
\]
for all such skeletons and all sufficiently large \(N\).
(As observed in Section~\ref{sec:ProofMain:FreeEnergy:LowerBound}, the minimal value and the minimizer of the dual problem depend continuously on the mass parameter; hence the stability constant can be chosen uniformly for masses sufficiently close to \(\alpha\).)

Choose \(\delta>0\) sufficiently small. On the one hand, by Lemma~\ref{lem:ProofMain:SkeletonUpperBound}, every such bad skeleton has weight at most
\[
  \Z{g}{L_N,A_N}[\bfv,\bftau] \leq \exp\{-N(m_* + c_{\epsilon/2}/3)\}
\]
for \(N\) large enough. Let us denote by \(\Z{g}{L_N,A_N}[d_{\mathrm{BL}}(\frm_\gamma^N,\frm_*)>\epsilon]\) the contribution to \(\Z{g}{L_N,A_N}\) due to paths \(\gamma\) such that \(d_{\mathrm{BL}}(\frm_\gamma^N,\frm_*)>\epsilon\). Since the number of skeletons is \(e^{\sfo(N)}\),
\[
  \Z{g}{L_N,A_N} \bigl[ d_{\mathrm{BL}}(\frm_\gamma^N,\frm_*)>\epsilon \bigr]
  \leq
  \exp\{-N(m_* + c_{\epsilon/2}/4)\}.
\]
On the other hand, the free-energy upper bound~\eqref{eq:ProofMain:UBFreeEnergy} gives
\[
  \Z{g}{L_N,A_N} \geq \exp\{-Nm_* - \sfo(N)\}.
\]
Taking the ratio yields the desired exponential bound.
\myQED

\subsubsection{Proof of Corollary~\ref{cor:TubeConcentration}}
It is enough to observe that a nearest-neighbor path cannot make a macroscopic deviation outside the \(\epsilon\)-neighborhood of \(\Gamma_*\) without placing a positive amount of monomer mass away from \(\Gamma_*\). Indeed, the function \(k \mapsto d_2(\gamma_k/N,\Gamma_*)\) changes by at most \(1/N\) at each step. Therefore, if \(N^{-1}\gamma\) leaves \(\calT_\epsilon^*\), then the path must spend at least \(c_\epsilon N\) steps at distance at least \(\epsilon/2\) from \(\Gamma_*\), for some \(c_\epsilon>0\). Equivalently,
\[
  \frm_\gamma^N(\bbR^2\setminus\calT_{\epsilon/2}^*) \geq c_\epsilon
\]
for all large \(N\). Taking a bounded Lipschitz function which vanishes on \(\calT_{\epsilon/2}^*\) and is positive outside
\(\calT_\epsilon^*\), we obtain \(d_{\mathrm{BL}}(\frm_\gamma^N,\frm_*)\geq c'_\epsilon\).
The claim thus follows from Theorem~\ref{thm:Concentration}.
\myQED

%%%%%%%%%%%%%%%%%%%%%%%
\section*{Acknowledgments}
%%%%%%%%%%%%%%%%%%%%%%%

The authors gratefully acknowledge the financial support of the Swiss National Science Foundation (SNSF) under grant number 200021\_219333. Both authors are members of the NCCR SwissMAP.

\bibliographystyle{plain}
\bibliography{PGF}

\end{document}